\documentclass[a4paper]{article}
\usepackage{amssymb,mathrsfs,amsmath,graphicx,amsthm}
\usepackage[left=2.5cm,right=2.5cm,top=2.5cm,bottom=2.5cm]{geometry}
\usepackage[font=sf, labelfont={sf,bf}, margin=1cm]{caption}

\long\def\symbolfootnote #1{\begingroup%
    \def\thefootnote{\fnsymbol{footnote}}\footnote[1]{#1}\endgroup}
\numberwithin{equation}{section}

\graphicspath{{pdf/}}

\newcommand{\ket}[1]{\left| #1 \right\rangle}
\newcommand{\bra}[1]{\left\langle #1 \right|}
\newcommand{\pic}[3]{\raisebox{#1}{\includegraphics[height=#2]{#3}}}
\newcommand{\picw}[3]{\raisebox{#1}{\includegraphics[width=#2]{#3}}}
\newcommand{\nn}{\nonumber}

\newcommand{\cotan}{{\rm cotan}}
\newcommand{\half}{\frac{1}{2}}
\newcommand{\id}{\mathbf{1}}
\newcommand{\On}{{\rm O}(n)}
\newcommand{\Pb}{\mathbb{P}}
\newcommand{\Rc}{\check{R}}

\newcommand{\ul}{\underline}
\newcommand{\wh}{\widehat}
\newcommand{\wt}{\widetilde}

\theoremstyle{plain}
\newtheorem{thm}{Theorem}[section]
\newtheorem{lm}[thm]{Lemma}
\newtheorem{prop}[thm]{Proposition}

\theoremstyle{definition}
\newtheorem{defn}{Definition}[section]

\theoremstyle{remark}

\newcommand{\lmref}[1]{Lemma~\ref{lm:#1}}
\newcommand{\propref}[1]{Proposition~\ref{prop:#1}}

\newcommand{\defnref}[1]{Definition~\ref{defn:#1}}

\newcommand{\secref}[1]{Section~\ref{sec:#1}}

\newcommand{\figref}[1]{Figure~\ref{fig:#1}}

\newcommand{\tabref}[1]{Table~\ref{tab:#1}}

\begin{document}

\title{Finite-size left-passage probability in percolation}

\author{
  Yacine Ikhlef\footnote{Yacine.Ikhlef@unige.ch}~
  and
  Anita Ponsaing\footnote{Anita.Ponsaing@unige.ch}
  \bigskip\\
      {\small \em
        \begin{minipage}{0.9\textwidth}
          \begin{center}
            Section de Math\'ematiques, University of Geneva, Switzerland
          \end{center}
        \end{minipage}
      }
}

\maketitle

\begin{abstract}
  We obtain an exact finite-size expression for the probability that a
  percolation hull will touch the boundary, on a strip of finite width. Our
  calculation is based on the $q$-deformed Knizhnik--Zamolodchikov approach,
  and the results are expressed in terms of symplectic characters. In the
  large size limit, we recover the scaling behaviour predicted by Schramm's
  left-passage formula. We also derive a general relation between the
  left-passage probability in the Fortuin--Kasteleyn cluster model and the
  magnetisation profile in the open XXZ chain with diagonal, complex boundary
  terms.
\end{abstract}

\section{Introduction}
\label{sec:intro}

Percolation models in two dimensions play an important role both in theoretical physics and mathematics. On the physics side, it was one of the first models where the Coulomb gas approach~\cite{DFSZ87,Nienhuis84} was used to predict the critical exponents~\cite{SaleurD87}, and where the concepts of boundary conformal field theory (CFT) were put in practice~\cite{Cardy92}. Nowadays, it still attracts the community's attention, especially for its relation to logarithmic CFT. On the mathematics side, many rigorous studies of percolation have been pursued~\cite{Werner09}, and Smirnov proved~\cite{Smirnov01} that site percolation on the triangular lattice has a conformally invariant scaling limit, described by Schramm--Loewner evolution (SLE) with $\kappa=6$. Also, in combinatorics, the Razumov--Stroganov relation~\cite{CantiniS11,RazStr04,RazStr05} identifies the components of the percolation transfer matrix eigenvector with the enumeration of plane partitions and alternating sign matrices.

The main objects of study in percolation are the percolation clusters and the lattice curves surrounding them, known as hulls. In the scaling limit, the correlation functions of these hulls are conjectured~\cite{SaleurD87} to be described by a Coulomb gas CFT~\cite{DFSZ87,Nienhuis84}, and thus to satisfy some partial differential equations (PDEs) given by the ``null-vector equations''. Some of these PDEs can be solved explicitly, e.g., the one for the crossing probability (the probability that a cluster connects two sides of a rectangle)~\cite{Cardy92}. A very fruitful approach to relate CFT and SLE is to express the null-state equations of CFT as martingale conditions for the SLE observables~\cite{BauerB03,Cardy05Network}.

The left-passage probability $P_{\rm left}(z)$, i.e., the probability for an open, oriented hull to pass to the left of a fixed point $z$ of the system, is one of these observables that can be easily obtained for percolation (and more generally, for the Potts and $\On$ models) both from the CFT and SLE viewpoints. In the SLE literature, this result is known as {\it Schramm's formula}~\cite{Schramm01}. In this paper, we address the determination of $P_{\rm left}$ {\it on the lattice}, in the infinite strip geometry, using rigorous techniques based on the Yang--Baxter and quantum Knizhnik--Zamolodchikov ($q$KZ) equations, as well as the Bethe Ansatz for the related six-vertex model.

The $q$KZ approach is particularly powerful for loop models with a trivial partition function ${\cal Z}=1$~\cite{dG05,dGPS09,dGP07,DF05,ZJ07}. In several cases, it allows the explicit determination of the dominant transfer matrix eigenvector, and it turns out that the components of this vector are integers enumerating plane partitions and alternating sign matrices. Also, this technique was recently used by one of the present authors to calculate a finite-size correlation function in percolation, namely the ``transverse current'' across a strip~\cite{dGNP10}. A complementary approach, with a larger scope, is to map a loop model onto an integrable spin chain~\cite{Bax82}, and use the Bethe Ansatz to obtain correlation functions in the form of determinants~\cite{KitanineKMNST07,KitanineKMNST08}. This approach actually extends to the Fortuin--Kasteleyn (FK) cluster model with cluster weight $Q$, and the percolation model is recovered for $Q=1$.

The layout of this paper is as follows. In \secref{perco}, we recall the exact equivalence~\cite{Bax82} between bond percolation on the square lattice and the Temperley--Lieb loop model with weight $n=1$, briefly review its conjectured relation to SLE$_6$, and state our main results. In \secref{qKZ}, we set up our notations for the transfer matrix and recall the basic steps of the $q$KZ approach. In \secref{Pb} we derive explicitly, for finite strip of width $L$, the probability $P_{\rm left}(z)$ with $z$ on the boundary of the strip. For a homogeneous system, we obtain fractional numbers with a simple combinatorial interpretation. In the large-$L$ limit, we recover the power law predicted by CFT and SLE. In \secref{Pleft}, we generalise to a generic point $z$: it turns out that similar symmetry and recursion relations hold, but are very difficult to solve in practice. However, we obtain two promising results in this case. First, we calculate $P_{\rm left}(z)$ numerically for homogeneous systems with up to $L=21$ sites, and observe good convergence to Schramm's formula. Second, we prove that, in the FK model, $P_{\rm left}(z)$ relates very simply to the magnetisation profile in an open XXZ spin chain with diagonal, complex boundary terms. We give our conclusions and perspectives in \secref{conclusion}.

\section{Percolation, Temperley--Lieb loops and SLE}
\label{sec:perco}

\subsection{Percolation hulls and their scaling exponents}

\begin{figure}[ht]
  \begin{center}
    \includegraphics[scale=0.8]{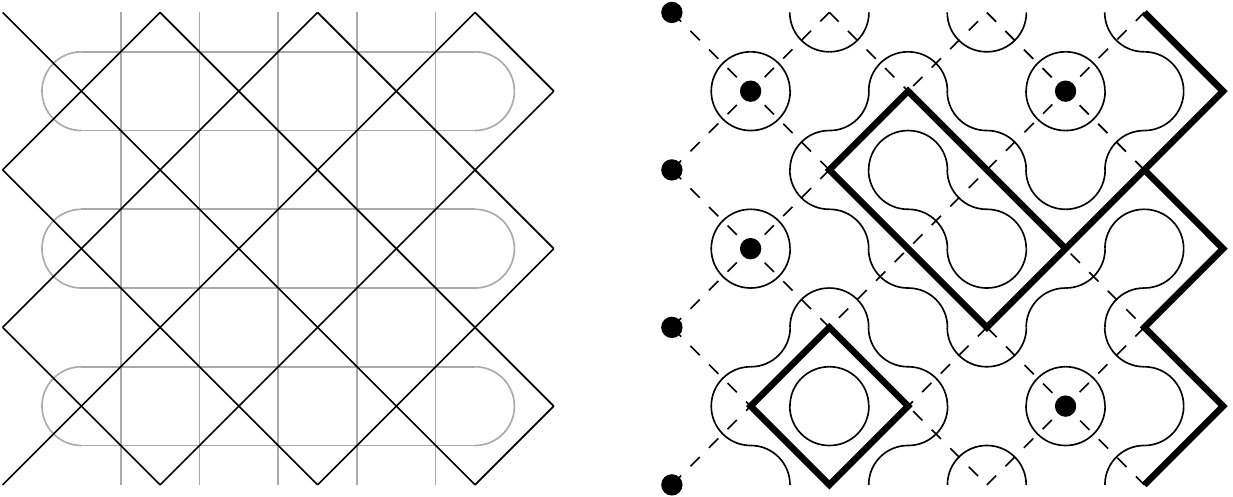}
  \end{center}
  \caption{Left: the original lattice $\cal L$ where bond percolation is
    defined (black) and its medial lattice $\cal M$ (grey). Right: an example
    cluster configuration (thick lines and dots) and the associated loop
    configuration (thin lines). On the left (resp. right) boundary, all edges are
    empty (resp. occupied).}
  \label{fig:medial}
\end{figure}

Consider a square lattice $\cal L$, on which each edge can be occupied by a bond with probability $p$, or empty with probability $(1-p)$. The connected components of the graph formed by all the sites and the occupied edges are called percolation clusters. We now look at the medial lattice $\cal M$ formed by the mid-edges of $\cal L$, where a loop configuration is associated to each cluster configuration~\cite{Bax82} (see \figref{medial}). These loops follow the external boundaries and the internal cycles of the clusters, and are called the percolation hulls. The model describing these loops is called the Temperley--Lieb loop model.

As usual, the scaling limit is defined by fixing a domain $\Omega$ of the plane, and covering it by a square lattice with spacing $a \to 0$. The scaling properties of percolation hulls at the critical point $p_c=1/2$ can be determined by the Coulomb gas approach, yielding the $\ell$-leg ``watermelon'' exponents $X_\ell = (\ell^2-1)/12$, the fractal dimension $d_f=7/4$, and the correlation-length exponent $\nu=4/7$.

\subsection{Strip geometry and the SLE model}

\begin{figure}[ht]
  \begin{center}
    \includegraphics[scale=0.8]{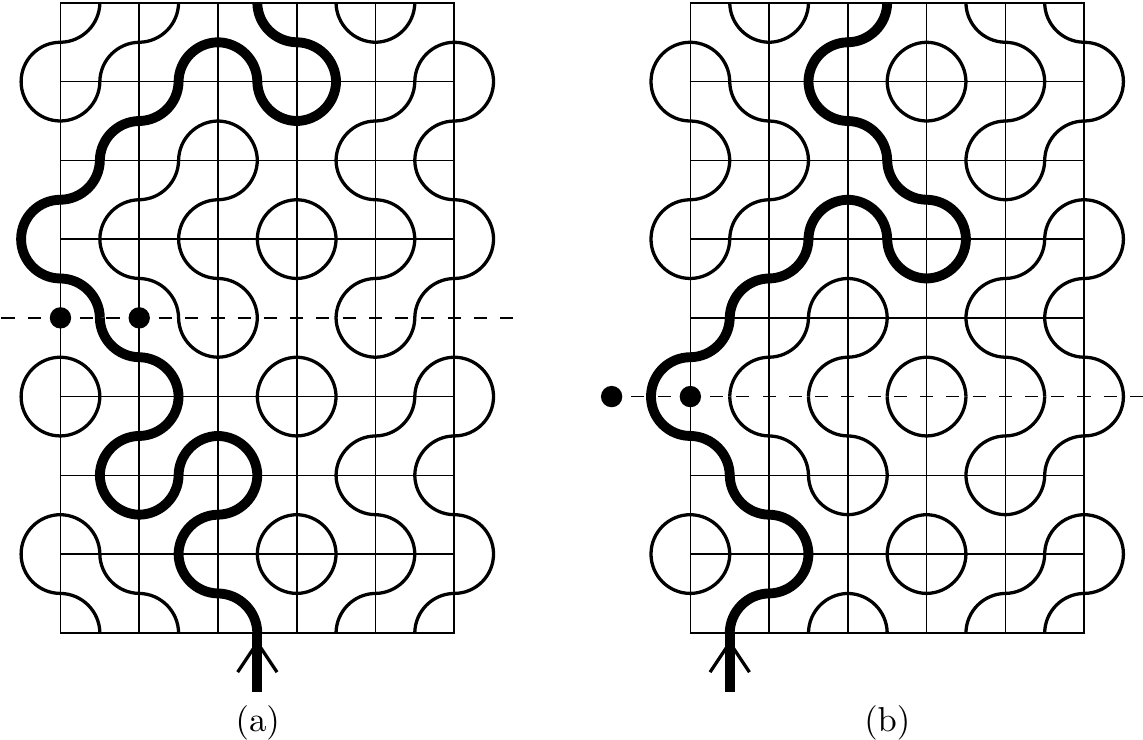}
  \end{center}
  \caption{(a) A configuration contributing to the boundary passage
    probability $P_b$ on the selected section (dotted line).
    (b) A configuration contributing to the boundary passage
    probability $\wh P_b$. In both cases, the dots indicate for which edge
    the passage probability is defined.}
  \label{fig:path}
\end{figure}

We suppose that the system is defined on an infinite strip of $\cal M$, of width $L$, where $L$ is odd. Let us specify the boundary conditions (BCs) which shall be used throughout the paper. On one boundary, we set all the edges of $\cal L$ to be occupied (wired BC), and on the other boundary, all edges are empty (free BC). This simply means the hulls must reflect on both boundaries. Since $L$ is odd, there exists an infinite open hull propagating along the strip, which we shall denote $\gamma$ (see \figref{path}).

In the scaling limit, the random curve $\gamma$ is conjectured to be distributed as chordal SLE$_{\kappa=6}$. In the SLE model, the left-passage probability $P_{\rm left}$ is obtained by solving an ordinary differential equation of order two. If we normalise the width of the strip to $L\times a=1$ and denote by $x \in [0,1]$ the horizontal coordinate, Schramm's formula reads~\cite{Schramm01}
\begin{equation} \label{eq:schramm}
  P_{\rm left}(x) = \half - \frac{\Gamma(4/\kappa)}
  {\sqrt{\pi} \ \Gamma\left( \frac{8-\kappa}{2\kappa}\right)}
  \ \cotan \ \pi x \times {_2F_1} \left(
  \half, \frac{4}{\kappa};\frac{3}{2};-\cotan^2 \pi x
  \right) \,,
\end{equation}
where ${_2F_1}$ is the hypergeometric function, which gives the probability of ``touching'' the boundary
\begin{equation} \label{eq:schramm2}
  P_{\rm left}(x) \ \mathop{\sim}_{x \to 0} \ \frac{\Gamma(4/\kappa)}
  {2\sqrt{\pi} \ \Gamma \left( \frac{8+\kappa}{2\kappa} \right)} \ (\pi x)^{\frac{8-\kappa}{\kappa}}
  \,,
\end{equation}
with $\kappa=6$ for percolation. In the present work, we derive some exact results for $P_{\rm left}(x)$ in the lattice model, i.e., we look for analogs of~\eqref{eq:schramm} and \eqref{eq:schramm2} {\it in finite size}.

\subsection{The Fortuin--Kasteleyn cluster model}

Most of the above results can be generalised to $4\leq \kappa \leq 8$ by considering a modified cluster model, called the Fortuin--Kasteleyn (FK) model, where each cluster gets a Boltzmann weight $Q$ (the critical regime is $0 \leq Q \leq 4$), so that the Boltzmann weight of a cluster configuration $C$ is
\[
  W[C] = Q^{\# \rm clusters(C)} \ v^{\# \rm occupied \ edges(C)} \,.
\]
Loops are defined on the medial lattice similarly to percolation, and, using the Euler relation, one finds that the above Boltzmann weight can be written as
\[
  W[C] \propto \sqrt{Q}^{\# \rm closed \ loops(C)}
  \ \left(\frac{v}{\sqrt{Q}}\right)^{\# \rm occupied \ edges(C)} \,.
\]
This defines the Temperley--Lieb loop model with weight $n=\sqrt{Q}$ (see \secref{TL}). It is conjectured (and proved for $Q=2$) that the hulls of FK clusters are distributed in the scaling limit as SLE$_\kappa$ with the relation
$$
\sqrt{Q} = -2 \cos \frac{4\pi}{\kappa} \,,
\qquad 4 \leq \kappa \leq 8 \,.
$$

\subsection{Statement of results}

\begin{itemize}

\item
 In the percolation model with $n=-(q+q^{-1})$, using the $q$KZ approach, we obtain explicitly (see \secref{inhomresult}) the probability that the open path $\gamma$ passes through a boundary edge. It reads:
  \begin{eqnarray*}
    P_b(z_1,z_2,\ldots,z_L) &=& \frac{\chi_{L-1}(z_2^2,\ldots,z_L^2)
      \ \chi_{L+1}(z_1^2,z_1^2,z_2^2,\ldots,z_L^2)}
    {\chi_L(z_1^2,\ldots,z_L^2)^2} \,, \\
    \wh P_b(w;z_1,z_2,\ldots,z_L) &=& \frac{(q^{-1}-q)(w^2-w^{-2})}{\prod_{i=1}^L k(1/w,z_i)}
    \times
    \frac{\chi_{L+1}(w^2,z_1^2,\ldots,z_L^2) \ \chi_{L+1}((q/w)^2,z_1^2,\ldots,z_L^2)}
         {\chi_{L}(z_1^2,\ldots,z_L^2)^2} \,,
  \end{eqnarray*}
 depending on the sublattice where the boundary edge sits (see \figref{path}). In the above expressions, $q$ has been set to $\exp(2i\pi/3)$, corresponding to $n=Q=1$, the $z_j$'s are the vertical spectral parameters, $w$ is the horizontal spectral parameter, $\chi_L$ is the symplectic character (see notations in \secref{qKZ}), and we have defined
  $$k(a,b):=(q/a)^2+(a/q)^2-b^2-b^{-2} \,.$$

\item
 For a homogeneous percolation system (see \secref{homresult}), this becomes
  $$
    P_b = \frac{A_V(L) \ A_V(L+2)}{N_8(L+1)^2} \,, 
    \qquad \text{and} \qquad
    \wh P_b = \frac{3}{4^L} \times \frac{A(L)^2}{N_8(L+1)^2 \ A_V(L)^2} \,,
  $$
 where $A(L)$, $A_V(L)$, and $N_8(L)$ are the number of $L\times L$ alternating sign matrices, $L\times L$ vertically symmetric alternating sign matrices, and $L\times L\times L$ cyclically symmetric self-complementary plane partitions respectively. Moreover, for a large system size $L$, both $P_b$ and $\wh P_b$ scale like $L^{-1/3}$ (see \secref{limresult}), which is consistent with~\eqref{eq:schramm2}.
  
\item
 In the critical Fortuin--Kasteleyn model with generic parameter $Q \in [0,4]$, we define the probability $X_j$ that the path $\gamma$ passes  through the $j$th horizontal edge of a given section as in \figref{path}a  (so that $P_b=X_1$). Defining $P_{\rm left}$ on the dual of $\cal M$, one can write
 $$
 P_{\rm left}\left( x_{j+1/2} \right) - P_{\rm left}\left( x_{j-1/2} \right)
 = (-1)^{j-1} X_j \,,
 \qquad \text{where} \qquad
 x_{j} := j/L \,.
 $$
 We find the relation (see \secref{FK})
 $$
 X_j = (-1)^{j-1} \  {\rm Re} \left(
 \frac{\bra{\Psi_0}\sigma_j^z\ket{\Psi_0}}
      {\langle \Psi_0|\Psi_0 \rangle}
      \right) \,,
 $$
 where $\ket{\Psi_0}$ is the groundstate eigenvector of the open XXZ Hamiltonian
 $$
 {\cal H}_{\rm XXZ} := \sum_{j=1}^{L-1} \left[
 \sigma_j^x \sigma_{j+1}^x + \sigma_j^y \sigma_{j+1}^y
 + \half(q+q^{-1}) \sigma_j^z \sigma_{j+1}^z
 \right]
 - \half(q-q^{-1}) (\sigma_1^z-\sigma_L^z) \,,
 $$
 where $\sqrt{Q} = -(q+q^{-1})$.
\end{itemize}

\section{The $q$KZ approach}
\label{sec:qKZ}

\subsection{The Temperley--Lieb loop model}
\label{sec:TL}

\begin{figure}[ht]
  \begin{center}
    \includegraphics[height=110pt]{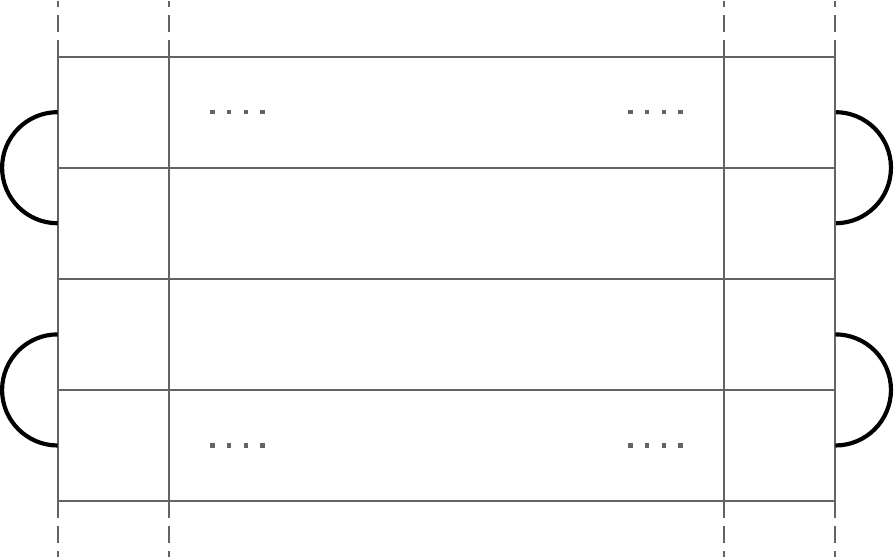}
    \caption{The lattice with two reflecting boundaries.}
    \label{fig:lat}
  \end{center}
  \vspace{-0.5cm}
\end{figure}

The Temperley--Lieb loop model with wired (or reflecting) boundaries \cite{dGP07,DF05,DF07,DFZJ07} is defined on a square lattice, where each face is decorated with loops in one of the following two ways:
\[\pic{-10pt}{25pt}{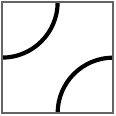}\quad \text{and} \quad\pic{-10pt}{25pt}{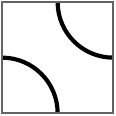}\quad.\]
Every closed loop gets a weight
$$
n = -(q+q^{-1}) \,.
$$
The chosen boundary conditions for this model require that the lattice is infinite in height and of finite width $L$. On the left and right are reflecting boundary conditions, as in \figref{lat}.

Drawing a horizontal line across the width of the lattice, we consider the connectivities of the loops below the line while ignoring the paths the loops take, as well as any closed loops. We refer to this pattern of connectivities as a link pattern, and we denote by ${\rm LP}_L$ the set of link patterns for a given system width $L$. For odd system size $L=2m-1$ they are enumerated by the $m$th Catalan number, $(2m)!/(m!(m+1)!)$. An example link pattern for $L=7$ is
\[ \pic{-13pt}{25pt}{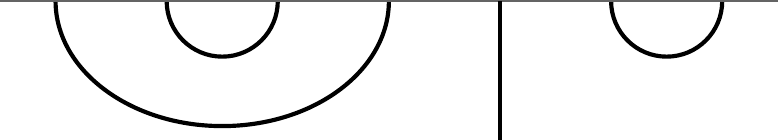} \,. \]
We label the link patterns by $\ket\alpha$, using the shorthand notation of ``$(\cdots)$'' to indicate a pair of sites connected by a loop, and ``$|$'' to indicate the single unpaired loop which always exists in an odd-sized system. As an example, the above link pattern is indicated by $\ket{\alpha} = \big|(())|()\big\rangle$.

The link patterns for a fixed $L$ form a representation of the Temperley--Lieb algebra, generated by $\{e_i,1\leq i\leq L-1\}$, with $e_i$ depicted as
\[ \pic{-11pt}{30pt}{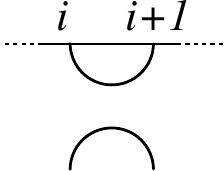} \,. \]
The relations for the Temperley--Lieb algebra are
\begin{equation}
  \label{eq:TL}
  \begin{array}{rcl}
    e_i^2 &=& n \ e_j \,, \\
    e_i e_{i\pm 1} e_i &=& e_i \,, \\
    e_i e_j &=& e_j e_i \qquad \text{if $|i-j|>1$.}
  \end{array}
\end{equation}
Multiplication corresponds to concatenating the depictions of the generators, giving a weight of $n$ to every closed loop, and disregarding the paths the loops take. In this way, we obtain relations between the link patterns such as $e_4\big|(())|()\big\rangle=\big||()()()\big\rangle$.

A state in $V_L={\rm span}({\rm LP}_L)$ is written as
\[ \ket{\phi}=\sum_{\alpha\in\text{LP}_L} \phi_\alpha \ket{\alpha}. \]
We look at all the possible configurations of two rows of the lattice, and consider how they send a given link pattern to another. We can write  this as a matrix $t$ which acts on $V_L$, and we refer to this as the transfer matrix.

We take an arbitrary initial state $\ket{\text{in}}$ and act $N$ times with the transfer matrix $t$. As $N\rightarrow\infty$, we get
\[ \lim_{N\rightarrow\infty}t^N\ket{\text{in}}\propto \Lambda^N\ket\Psi, \]
where $\Lambda$ is the maximum eigenvalue of $t$ and $\ket\Psi$ is the corresponding eigenvector, also known as the ground state. When $n=1$ all the weights are probabilities and therefore $\Lambda=1$. The components of $\ket\Psi$ can be thought of as the relative probabilities of the possible link patterns. In a similar way we define $\bra\Psi$, which is the groundstate of the rotated lattice, giving the relative probabilities of upward link patterns. The inner product between upward and downward link patterns is simply
\[ \bra\beta\alpha\rangle := n^{\#\text{closed loops}} , \quad \forall\alpha,\beta. \]
Hence, the expectation value of some observable $\cal O$ reads
\[
\langle {\cal O} \rangle := \frac{\bra{\Psi} {\cal O} \ket{\Psi}}
          {\langle \Psi \ket{\Psi}} \,.
\]

\subsection{The $R$-matrix}
The possible states of each lattice square are described by the $R$-matrix.
\begin{defn}
\label{defn:R}
  \[ R(w,z)=\pic{-25pt}{50pt}{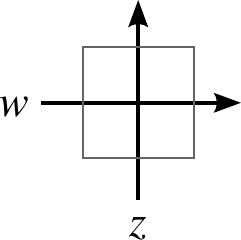}=
  \frac{[qz/w]}{[qw/z]}\quad\pic{-10pt}{25pt}{R1}\;
  +\;\frac{[z/w]}{[qw/z]}\quad\pic{-10pt}{25pt}{R2}\ ,\]
  where
  $$
  [z] := z-1/z \,.
  $$
\end{defn}
\begin{lm}
This $R$-matrix satisfies the Yang--Baxter equation (YBE)
\begin{equation}
\label{eq:YBE}
\pic{-35pt}{70pt}{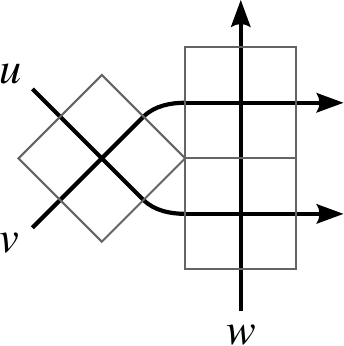}\;=\;\pic{-35pt}{70pt}{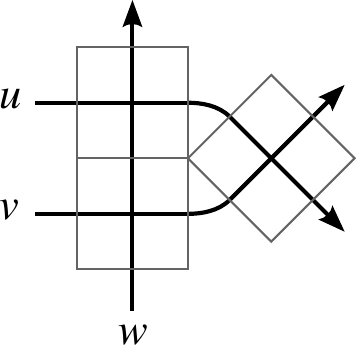}\;,
\end{equation}
the unitarity relation
\begin{equation}
\label{eq:unitarity}
\pic{-18pt}{40pt}{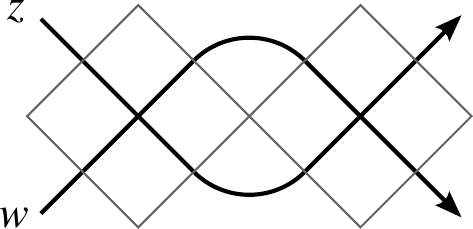}\;=\;\pic{-18pt}{40pt}{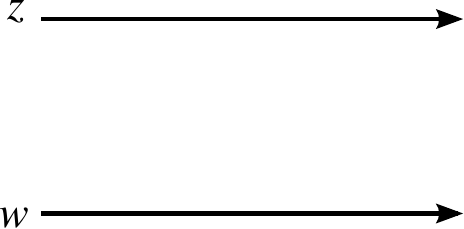}\;,
\end{equation}
and the crossing relation
\begin{equation}
\label{eq:cross}
\pic{-30pt}{60pt}{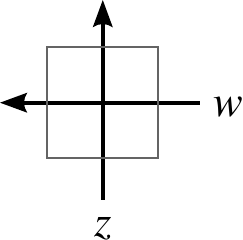}\;=\;\pic{-30pt}{60pt}{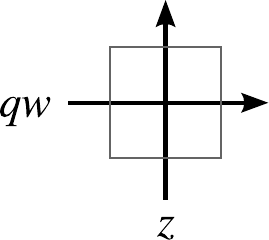}\;.
\end{equation}
\end{lm}

\begin{defn}
  The corresponding operator, acting on $V_L$, is
  \begin{equation}
    \label{eq:Rc}
    \Rc_j(w) := \frac{[q/w]}{[qw]} \ \id - \frac{[w]}{[qw]} \ e_j \,.
  \end{equation}
\end{defn}

\subsection{The transfer matrix, symmetries and recursions}

We now define the transfer matrix $t$, which describes all the possible configurations of two lattice rows \cite{DF05,Skl88}.
\begin{defn}
  \[ t(w;z_1,\ldots,z_L)= {\rm Tr}_w \left[
    R(w,z_1)\ldots R(w,z_L)R(z_L,1/w)\ldots R(z_1,1/w)
    \right], \]
  or pictorially,
  \[ t(w;z_1,\ldots,z_L)=\pic{-30pt}{60pt}{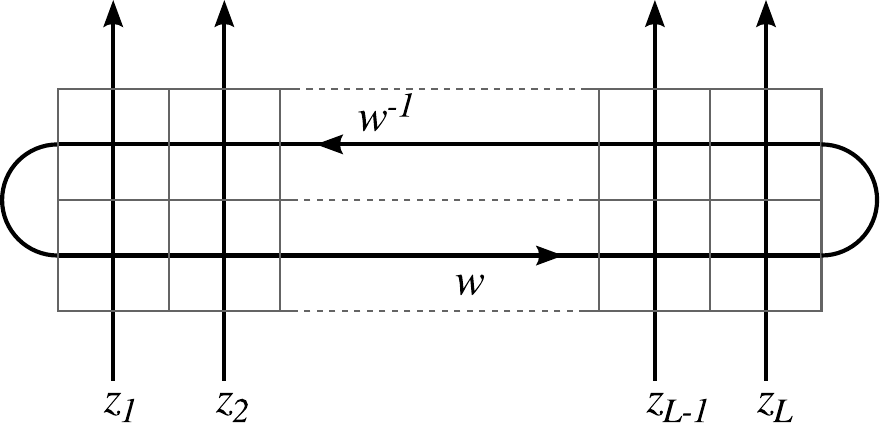}. \]
\end{defn}

\begin{lm}
  Thanks to the YBE \eqref{eq:YBE}, the transfer matrix satisfies the interlacing relation
  \begin{equation}
    \label{eq:interlace}
    \Rc_i(z_i/z_{i+1}) t(w;z_i,z_{i+1})=t(w;z_{i+1},z_i) \Rc_i(z_i/z_{i+1}),
  \end{equation}
  pictorially,
  \[ \pic{-31pt}{70pt}{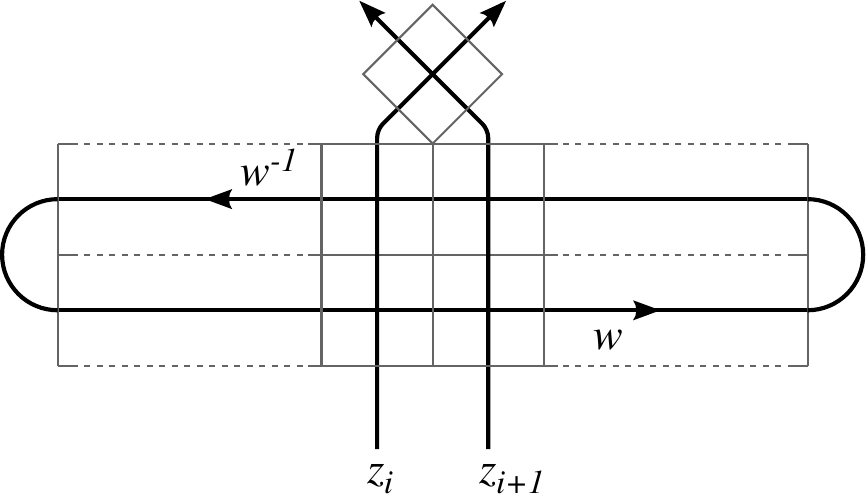}\quad=\quad\pic{-37pt}{70pt}{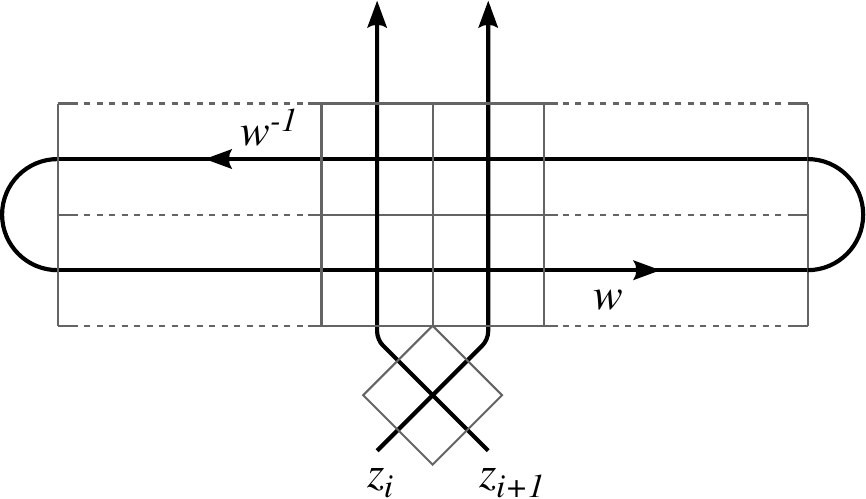}\quad. \]
  Considering the possible configurations of the two tiles at either position $1$ or position $L$ of the transfer matrix also gives us, respectively,
  \begin{equation}
    \label{eq:boundinterlace}
    \begin{split}
      t(w;z_1,z_2\ldots) &= t(w;1/z_1,z_2,\ldots),\\
      t(w;\ldots,z_{L-1},z_L) &= t(w;\ldots,z_{L-1},1/z_L).
    \end{split}
  \end{equation}
\end{lm}

\begin{lm}
  By acting the transfer matrix on a small link from site $i$ to $i+1$ (denoted by $\varphi_i$) and setting $z_{i+1}=qz_i$ we find the relation
  \begin{equation}
    \label{eq:Trecur}
    t_L(z_i,z_{i+1}=qz_i)\circ\varphi_i = \varphi_i \circ t_{L-2}(\hat z_i,\hat z_{i+1}),
  \end{equation}
  where $\hat z$ means that $z$ is missing from the list of arguments.
\end{lm}
\begin{proof}
  Considering first the bottom row, we use the crossing relation \eqref{eq:cross} and then the unitarity relation \eqref{eq:unitarity}:
  \[ \picw{-23pt}{70pt}{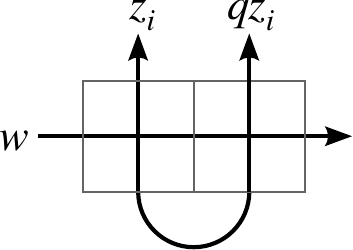}\qquad=\qquad\picw{-23pt}{70pt}{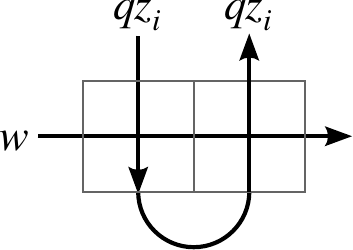}\qquad=\qquad\picw{-3pt}{70pt}{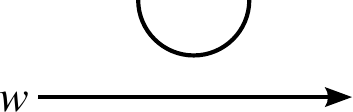}\;, \]
  and see that the bottom row no longer depends on $z_i$ and $z_{i+1}$. We repeat the procedure for the top row:
  \[ \picw{-23pt}{70pt}{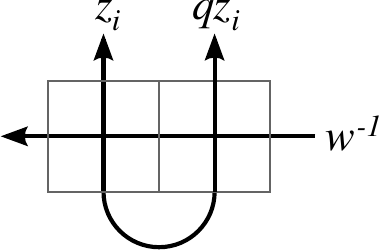}\qquad=\qquad\picw{-23pt}{70pt}{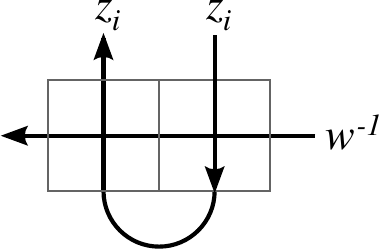}\qquad=\qquad\picw{-4pt}{70pt}{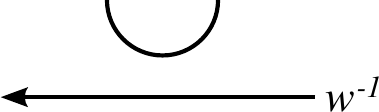}\;. \]
  and the result follows.
\end{proof}

\subsection{The dominant eigenvector}

It is possible to show (see for example \cite{DF05}) that two copies of the transfer matrix with different values of the parameter $w$ commute, and therefore that the groundstate eigenvector does not depend on $w$. Explicitly, the eigenvalue equation thus becomes
\begin{equation}
  \label{eq:TPsi}
  t(w;z_1,\ldots,z_L)\ket{\Psi(z_1,\ldots,z_L)}=\ket{\Psi(z_1,\ldots,z_L)} \,,
\end{equation}
with the ground state eigenvector given by
\[ \ket{\Psi(z_1,\ldots,z_L)}
=\sum_{\alpha\in\text{LP}_L}\psi_\alpha(z_1,\ldots,z_L)\ket\alpha. \]
From the expression of the $R$-matrix, the coefficients in the eigenvalue equation~\eqref{eq:TPsi} are all rational functions of the $z_j$'s, and hence one can normalise $\ket{\Psi(z_1,\ldots,z_L)}$ so that all the components $\psi_\alpha(z_1,\ldots,z_L)$ are Laurent polynomials in the $z_j$'s. Moreover, one requires that these components have no common factor.

With this normalisation, the interlacing relations~\eqref{eq:interlace} and \eqref{eq:boundinterlace} yield the $q$-deformed Knizhnik--Zamolodchikov equation for the ground state eigenvector, expressed in the form
\begin{align*}
  \Rc_i(z_i/z_{i+1})\ket{\Psi(z_1,\ldots,z_L)}&=\pi_i\ket{\Psi(z_1,\ldots,z_L)},\\
  \ket{\Psi(z_1,\ldots,z_L)}&=\ket{\Psi(1/z_1,\ldots,z_L)},\\
  \ket{\Psi(z_1,\ldots,z_L)}&=\ket{\Psi(z_1,\ldots,1/z_L)},
\end{align*}
where $\pi_i f(z_i,z_{i+1})=f(z_{i+1},z_i)$.

Acting on $\ket{\Psi_{L-2}(\hat z_i,\hat z_{i+1})}$ with both sides of \eqref{eq:Trecur}, we get
\[ t_L(z_i,z_{i+1}=qz_i)\ \varphi_i\ket{\Psi_{L-2}(\hat z_i,\hat z_{i+1})}
= \varphi_i\ket{\Psi_{L-2}(\hat z_i,\hat z_{i+1})}.\]
Since the ground state is unique, $\ket{\Psi_L(z_{i+1}=qz_i)}$ and $\varphi_i\ket{\Psi_{L-2}(\hat z_i,\hat z_{i+1})}$ are linearly related, and one can show that $\ket{\Psi}$ satisfies the recursion relation
\begin{equation}
  \label{eq:Psirecur}
  \ket{\Psi_L(z_{i+1}=qz_i)} =
  \varphi_i\ket{\Psi_{L-2}(\hat z_i,\hat z_{i+1})}
  \times (-1)^L \prod_{j\notin \{i,i+1\}} k(z_i,z_j) \,,
\end{equation}
with
$$ k(a,b) =[qb/a] \ [q/ab] \,. $$


\subsection{Solution for the eigenvector}
The $q$KZ equation forces certain symmetry requirements on the components of $\ket\Psi$, which lead to the solution for the ground state. For instance,
\[ \psi_{|(\cdots()\cdots)}=(-1)^{\frac{L}2(\frac{L}2+1)}\prod_{1\leq i<j\leq \frac{L+1}2}k(z_j,z_i)\prod_{\frac{L+3}2\leq i<j\leq L}k(1/z_i,z_j). \]
By considering the Dyck path representation of the link patterns, one can write the other components in terms of factorised operators acting on this component. We will not give the explicit solution here as it is not needed for our calculations. A full explanation of the procedure is in Section~4.1 of \cite{dGP07}.

\subsection{The normalisation factor $Z_L$}

The normalisation factor $Z_L$ is defined as
\[
  Z_L(z_1,\ldots,z_L) := \sum_{\alpha\in\text{LP}_L}\psi_\alpha(z_1,\ldots,z_L) \,.
\]
To express $Z_L$, we first define the completely symmetric polynomial character of the symplectic group.
\begin{defn}
The symplectic character $\chi_{\lambda}$ associated to a partition $\lambda$ is given by
\[
\chi_{\lambda}^{(L)} (u_1,\ldots ,u_L)
=\frac{\det{\left[u_i^{\mu _j}-u_i^{-\mu _j}\right]}}
{\det{\left[u_i^{\delta_j}-u_i^{-\delta_j}\right]}} \,,
\]
where $\delta_j= L-j+1$ and $\mu_j=\lambda_j+\delta_j$.
\end{defn}
Throughout this paper, we shall restrict to the partition $\lambda_j=\left\lfloor\frac{L-j}{2}\right\rfloor$, for which $\chi$ has the special recursion
\begin{equation}
  \label{eq:chirecur}
  \chi_{\lambda}^{(L)} (u_1^2,\ldots ,u_L^2)|_{u_i=qu_j}
  =(-1)^L \prod_{\ell\neq i,j}k(u_j,u_\ell)
  \ \chi_{\lambda}^{(L-2)} (\ldots,\hat u_i^2,\ldots,\hat u_j^2,\ldots).
\end{equation}
We will use the shorthand notation $\chi_L(\ldots) := \chi_{\lambda}^{(L)}(\ldots)$, with the particular choice of $\lambda$ given above.

\begin{prop}
The normalisation $Z_L$ is given by
\[ Z_L=\chi_L(z_1^2,\ldots,z_L^2). \]
\end{prop}
\begin{proof}
The recursive property \eqref{eq:Psirecur} of the ground state eigenvector is easily extended to its components, and thus to the normalisation,
\begin{equation}
  \label{eq:Zrecur}
  Z_L(z_{i+1}=qz_i)=Z_{L-2}(\hat z_i,\hat z_{i+1})
  \times (-1)^L \prod_{\ell\neq i,i+1}k(z_i,z_\ell).
\end{equation}
As $Z_L$ is a symmetric function (easily proven using \eqref{eq:unitarity} and the $q$KZ equation), this can be generalised to
\[
  Z_L(z_j=qz_i)=Z_{L-2}(\hat z_i,\hat z_j) \times (-1)^L \prod_{\ell\neq i,j}k(z_i,z_\ell).
\]
The symplectic character $\chi_L(z_1^2,\ldots,z_L^2)$ also satisfies these recursions, and it is straightforward to show that these recursions are enough to satisfy the degree of $Z_L$, which is set by solving the $q$KZ equation. It remains to show that the statement is true for a small system size, which is done by observing that for $L=1$ both the left and the right hand side must be $1$.\symbolfootnote{Note that this proof is only valid for odd $L$; for even $L$ we must prove the statement for $L=2$ as well. We omit this part of the proof as we are only interested in odd system sizes.}
\end{proof}

\section{Boundary passage probabilities}
\label{sec:Pb}

When $L$ is odd, all link patterns have an unpaired odd site. In the lattice this site belongs to an open path extending from $-\infty$ to $\infty$. In this section we will calculate two probabilities: $P_b$, the probability that this infinite path passes through the first site at a given vertical position (\figref{path}a); and $\wh P_b$, the probability that this loop passes through the left boundary at a given vertical position (\figref{path}b).

\subsection{Definitions}

We first define $\bra{\Psi}$ to be the ground state eigenvector of the rotated system, given by
\[ \bra{\Psi}=\sum_{\alpha \in \text{LP}_L}\bar\psi_{\bar\alpha}\bra{\bar\alpha}, \]
and related to $\ket\Psi$ by
\[ \bar\psi_{\bar\alpha}(z_1,\ldots,z_L)=\psi_\alpha(z_L,\ldots,z_1), \]
where $\alpha$ and $\bar\alpha$ are related by a rotation of $\pi$.

\begin{defn}
The first site passage probability is given by
\begin{equation}
\label{eq:P1}
P_b^{(L)}=\frac{\bra\Psi\rho\ket\Psi}{\bra\Psi\Psi\rangle},
\end{equation}
where $\rho$ acts between a link pattern and a rotated link pattern,
\[ \bra\beta\rho\ket\alpha, \]
giving 1 if the open path formed by these two link patterns goes through the first site, and 0 if it does not. For example, the following configuration gives a weight of 1 ($\rho$ is depicted as two dots marking the first site):
\[ \pic{0pt}{30pt}{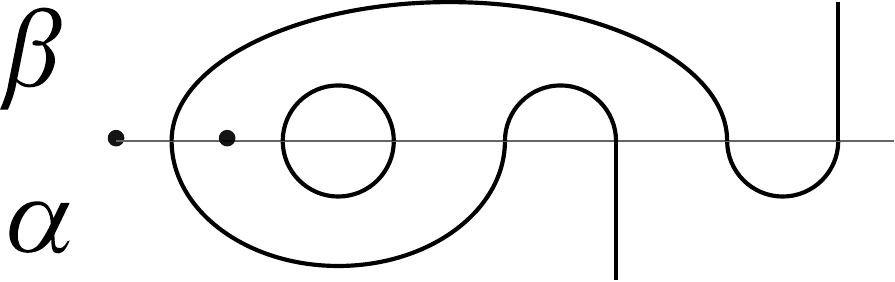} \]
\end{defn}

\begin{defn}
The boundary passage probability $\wh P^{(L)}$ is given by
\begin{equation}
\label{eq:P0}
\wh P_b^{(L)}=\frac{\bra\Psi\wh\rho\ket\Psi}{\bra\Psi\Psi\rangle},
\end{equation}
where $\wh\rho$ marks out the left boundary loop in the transfer matrix. It is depicted as
\[ \pic{0pt}{70pt}{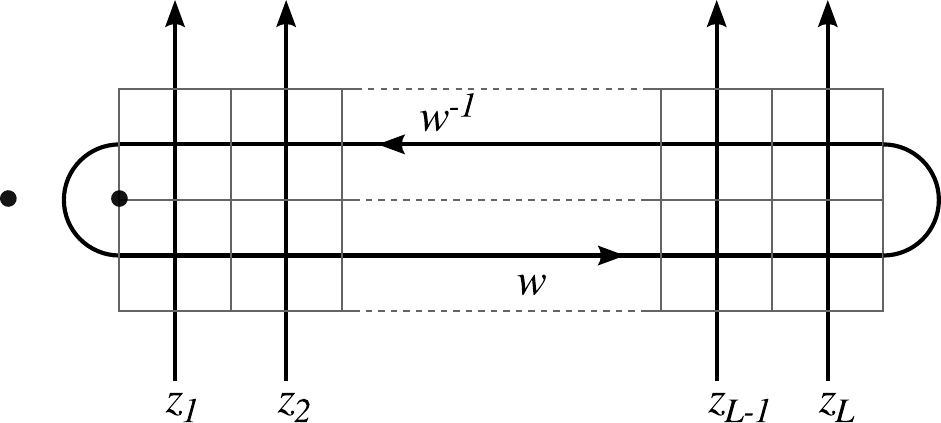} \]
and like $\rho$, acts between and upward and a downward link pattern, multiplying each term in the transfer matrix by 1 if the infinite loop passes through the left boundary loop, and 0 if it does not. For example, the following configuration is multiplied by 1:
\[ \pic{0pt}{85pt}{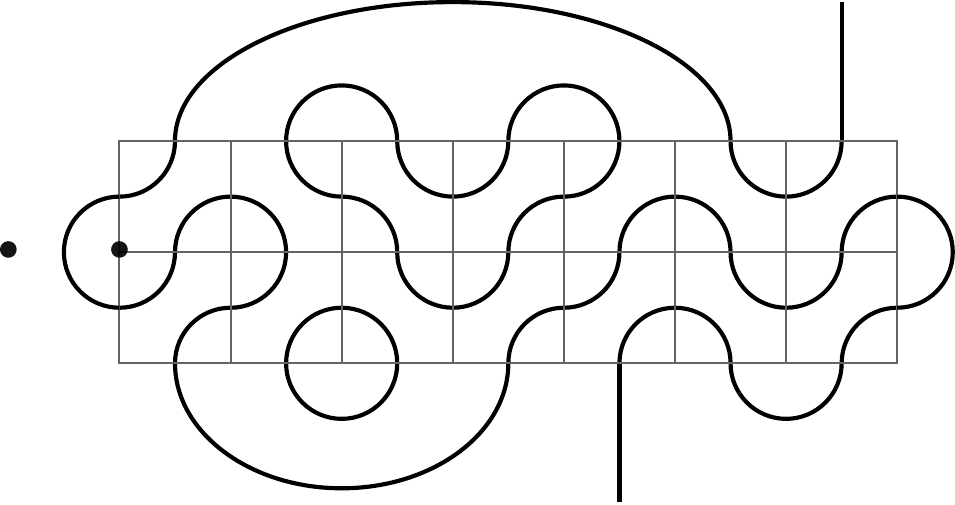} \]
\end{defn}

As $\bra\beta\alpha\rangle=1$, $\forall \alpha,\beta$, the denominator of $P_b^{(L)}$ and $\wh P_b^{(L)}$ becomes
\[
\bra\Psi\Psi\rangle
= \sum_{\alpha,\beta} \bar\psi_{\bar\beta} \psi_\alpha \bra\beta\alpha\rangle
= Z_L(z_L, \dots, z_1) \ Z_L(z_1, \dots, z_L)
= [Z_L(z_1, \dots, z_L)]^2 \,.
\]

\subsubsection*{Example: $L=3$}
For $L=3$ there are two link patterns, $|()$ and $()|$. Solving the $q$KZ equation for the components of the eigenvector gives
\begin{align*}
\psi_{|()}(z_1,z_2,z_3)&=k(z_2,z_1),\\
\psi_{()|}(z_1,z_2,z_3)&=\frac{[qz_2/z_3]}{[z_2/z_3]}\left(1-\pi_2\right)\psi_{|()}\\
&=k(1/z_2,z_3),
\end{align*}
and $Z_3=\psi_{|()}+\psi_{()|}$ is simply $\chi_3(z_1^2,z_2^2,z_3^2)$.

The only combination of upward and downward link patterns which does not contribute to $P_b^{(3)}$ is $\alpha=\bar\beta=()|$. We thus find the first site passage probability by brute force as
\begin{align*}
  P_b^{(3)}&= \frac1{Z_3^2} \left(
  \psi_{|()}\bar\psi_{|()}+\psi_{()|}\bar\psi_{|()}+\psi_{|()}\bar\psi_{()|}
  \right) \\
  &= \frac1{Z_3^2} \left( 
  \psi_{|()}(z_1,z_2,z_3)\psi_{()|}(z_3,z_2,z_1)+\psi_{()|}(z_1,z_2,z_3)\psi_{()|}(z_3,z_2,z_1)
  \right. \\
  &\qquad \qquad \left.+\psi_{|()}(z_1,z_2,z_3)\psi_{|()}(z_3,z_2,z_1)
  \right) \\
  &=\frac1{Z_3^2} \left(
  5+z_1^4+z_1^{-4}+2(z_1^2z_2^2+z_1^2z_2^{-2}+z_1^{-2}z_2^2+z_1^{-2}z_2^{-2}+z_1^2z_3^2
  \right.\\
  &\left.\qquad +z_1^2z_3^{-2}+z_1^{-2}z_3^2+z_1^{-2}z_3^{-2})
  +z_2^2z_3^2+z_2^2z_3^{-2}+z_2^{-2}z_3^2+z_2^{-2}z_3^{-2}\right).
\end{align*}


\subsection{Symmetries}

\begin{prop}
  \label{prop:P1symm}
  $P_b^{(L)}$ is symmetric in $z_2,\ldots,z_L$ and invariant under $z_i\rightarrow 1/z_i$ for $i\geq 2$.
\end{prop}
\begin{proof}
  This proof is similar to the proof that $Z_L$ is symmetric, and uses the $q$KZ equation. Remembering that an $\Rc$-matrix acting between sites $i\neq 1$ and $i+1$ commutes with $\rho$, we insert the identity $\Rc_i(z_{i+1}/z_i)\Rc_i(z_i/z_{i+1})$ into
  the definition for $P_b^{(L)}$:
  \begin{align*}
    P_b^{(L)}(\dots z_i,z_{i+1} \dots)
    &= \frac{\bra{\Psi_L(z_i,z_{i+1})}\Rc_i(z_{i+1}/z_i)\ \rho\ \Rc_i(z_i/z_{i+1})\ket{\Psi_L(z_i,z_{i+1})}}
    {Z_L(z_i,z_{i+1})^2} \\
    &= \frac{\bra{\Psi_L(z_{i+1},z_i)}\rho\ket{\Psi_L(z_{i+1},z_i)}}
    {Z_L(z_{i+1},z_i)^2} \\
    &= P_1^{(L)}(\dots z_{i+1},z_i \dots) \,.
  \end{align*}
  The invariance of $P_b^{(L)}$ under $z_i\rightarrow 1/z_i$ simply follows from the invariance of $\ket{\Psi_L}$ and $\bra{\Psi_L}$ under $z_L\rightarrow 1/z_L$, as well as the above symmetry.
\end{proof}

\begin{prop}
  \label{prop:P0symm}
  $\wh P_b^{(L)}$ is symmetric in all the $z_i$'s and invariant under $z_i\rightarrow 1/z_i$, $\forall i \in \{ 1, \dots, L \}$.
\end{prop}
\begin{proof}
  The proof is similar to the previous proof, however it uses the fact that the interlacing condition \eqref{eq:interlace} for the transfer matrix is also satisfied by $\wh\rho$.
\end{proof}

\subsection{Recursions}

\begin{prop}
  \label{prop:P1rec}
  $P_b^{(L)}$ satisfies the following recursions:
  \begin{equation}
    \label{eq:P1rec}
    P_b^{(L)}|_{z_L^2=(qz_k)^{\pm 2}}=P_b^{(L-2)}(\hat z_k,\hat z_L),\qquad 1<k<L.
  \end{equation}
\end{prop}
\begin{proof}
  The denominator of $P_b^{(L)}$ has the recursion \eqref{eq:Zrecur}
  \[ Z_L(z_L^2=q^2z_{L-1}^2)^2=\prod_{i=1}^{L-2}k(z_{L-1},z_i)^2\ Z_{L-2}(z_1,\ldots,z_{L-2})^2, \]
  and we will show that the numerator has the same recursion factor.

  From the recursion of the right eigenvector we have
  $$
  \bra{\Psi_L}\rho\ket{\Psi_L}|_{z_L^2=(qz_{L-1})^2}=
  \prod_{i=1}^{L-2}k(z_{L-1},z_i)\bra{\Psi_L(z_L^2=q^2z_{L-1}^2)}\rho
  \ \varphi_{L-1}\ket{\Psi_{L-2}(z_1,\ldots,z_{L-2})},
  $$
  and since $\rho$ commutes with $\varphi_{L-1}$, we can consider $\bra{\Psi_L}\varphi_{L-1}$, which is the $\pi$ rotation of the vector $\varphi_1^\dag\ket{\Psi_L(z_L,\ldots,z_1)}$. Here, $\varphi_1^\dag$ is the bottom half of the TL operator $e_1$, sending a link pattern of size $L$ to one of size $L-2$. We can thus obtain our desired result by calulating $e_1\ket{\Psi_L(z_L,\ldots,z_1)}$ and removing the resulting link from site 1 to site 2. Here the $q$KZ equation comes in useful, as
  \begin{align*}
    &\varphi_1 \varphi_1^\dag \ket{\Psi_L(z_L,\ldots,z_1)}|_{z_L^2=(qz_{L-1})^2} \\
    &\qquad\qquad = e_1\ket{\Psi_L(z_L,\ldots,z_1)}|_{z_L^2=(qz_{L-1})^2} \\
    &\qquad\qquad = -\left.\left(\frac{[z_{L-1}/qz_L]}{[z_{L-1}/z_L]}\pi_{L-1}
    +\frac{[z_L/qz_{L-1}]}{[z_L/z_{L-1}]}\right)\ket{\Psi_L(z_L,\ldots,z_1)}\right|_{z_L^2=(qz_{L-1})^2}
    \\
    &\qquad\qquad = -\frac{[q]}{[1/q]}\ket{\Psi_L(z_{L-1},qz_{L-1},\ldots,z_1)} \\
    &\qquad\qquad = \varphi_1\ket{\Psi_{L-2}(z_{L-2},\ldots,z_1)}
    \times \prod_{i=1}^{L-2}k(z_{L-1},z_i).
  \end{align*}
  Therefore, we have
  \[
  \varphi_1^\dag\ket{\Psi_L(z_L,\ldots,z_1)}|_{z_L^2=(qz_{L-1})^2}=\ket{\Psi_{L-2}(z_{L-2},\ldots,z_1)}
  \times \prod_{i=1}^{L-2}k(z_{L-1},z_i),
  \]
  which is the $\pi$ rotation of $\bra{\Psi_{L-2}} \times \prod_{i=1}^{L-2}k(z_{L-1},z_i)$, and thus,
  \[
  \bra{\Psi_L}\rho\ket{\Psi_L}|_{z_L^2=(qz_{L-1})^2}
  =\prod_{i=1}^{L-2}k(z_{L-1},z_i)^2\bra{\Psi_{L-2}}\rho_1\ket{\Psi_{L-2}}.
  \]
  It follows that $P_b^{(L)}|_{z_L^2=(qz_{L-1})^2}=P_b^{(L-2)}$.
  
  The other relations in \eqref{eq:P1rec} follow from the invariance of $P_b$ under $z_i\leftrightarrow z_j$ and under $z_i\rightarrow 1/z_i$, for $i,j\neq 1$.
\end{proof}

\begin{prop}
\label{prop:P0rec}
$\wh P_b^{(L)}$ satisfies the following recursions:
\begin{equation}
\label{eq:P0rec}
\wh P_b^{(L)}|_{z_L^2=(qz_k)^{\pm 2}}=\wh P_b^{(L-2)}(\hat z_k,\hat z_L),\qquad 1\leq k<L.
\end{equation}
\end{prop}
\begin{proof}
This proof is very similar to the previous one, but it also relies on the fact that $\wh \rho$ satisfies the same recursion as the transfer matrix \eqref{eq:Trecur},
\[
\wh\rho^{(L)}(z_{i+1}=qz_i) \circ \varphi_i
= \varphi_i \circ \wh\rho^{(L-2)}(\hat z_i,\hat z_{i+1}).
\qedhere \] \,.
\end{proof}

\subsection{Exact solution}

\subsubsection{Inhomogeneous system}
\label{sec:inhomresult}
\begin{prop}
  The explicit formula for the first site passage probability is
  \begin{equation}
    \label{eq:P1form}
    P_b^{(L)} =
    \frac{\chi_{L-1}(z_2^2,\ldots,z_L^2)\chi_{L+1}(z_1^2,z_1^2,z_2^2,\ldots,z_L^2)}
         {\chi_L(z_1^2,\ldots,z_L^2)^2}.
  \end{equation}
\end{prop}

\begin{proof}
  There are three steps to this proof. First, the degree of the proposed expression must be shown to agree with the definition. Second, enough recursions (or values of the polynomial at specified points) must be found to satisfy the degree. Thirdly, the proposed expression must be shown to be true for a small size example ($L=1$), to initialise the recursion.

  The denominator of \eqref{eq:P1} can be easily shown to be $Z_L^2$, which has a degree of $L-1$ in each variable $z_i^2$, because $L$ is always odd. The action of $\rho$ can not raise the polynomial degree, so the numerator of $P_b^{(L)}$ must have at most the same degree as the denominator. The degree of the numerator of \eqref{eq:P1form} is $(L-3)/2+(L-1)/2=L-2$ in each $z_i^2$.\symbolfootnote{The definition of the transfer matrix implies that the components of the eigenvector, and thus the passage probabilities, are only functions of $z_i^2$ and do not depend on any odd powers of $z_i$. This fact is crucial to these proofs.}

  We view $P_b^{(L)}$ as a polynomial in $z_L^2$ with coefficients in the other $z_i$'s, and to satisfy the degree we need to know the value of the polynomial for at least $2L-1$ values of $z_L^2$. The following recursions, which come from \propref{P1rec}, will give the value of $P_B^{(L)}$ at $4(L-2)$ values of $z_L^2$:
  \[
  P_b^{(L)}|_{z_L^2=(qz_k)^{\pm 2}} = P_b^{(L-2)}(\hat z_k,\hat z_L),\qquad 1<k<L.
  \]
  The proof that these recursions are satisfied by the proposed expression for $P_b^{(L)}$ is straightforward, based on the known recursions for $\chi$ \eqref{eq:chirecur}.

  It thus remains to show that the expression is true for $L=1$, which is easy to do as $P_b^{(1)}$ must trivially be 1, and $\chi_0=\chi_1=\chi_2=1$.
\end{proof}

\begin{prop}
  The explicit formula for the boundary passage probability is
  \begin{equation}
    \label{eq:P0form}
    \wh P_b^{(L)} =
    \frac{-[q][w^2]}{\prod_{i=1}^L k(1/w,z_i)}
    \times
    \frac{\chi_{L+1}(w^2,z_1^2,\ldots,z_L^2)\chi_{L+1}((q/w)^2,z_1^2,\ldots,z_L^2)}
         {\chi_{L}(z_1^2,\ldots,z_L^2)^2}.
  \end{equation}
\end{prop}

\begin{proof}
  This is a similar proof to the previous one, with the same three steps.

  The denominator of \eqref{eq:P0} is $Z_L^2$ multiplied by the denominator of $\wh\rho$, which is the same as the denominator of the transfer matrix; that is, $\prod_{i=1}^L k(1/w,z_i)$. Thus the degree of the denominator is $L$ in each $z_i^2$. Again because of the definition of $\wh\rho$ the numerator of $\wh P_b^{(L)}$ will be at most the same as the denominator. The degree of the numerator of \eqref{eq:P0form} is $2(L-1)/2=L-1$ in each $z_i^2$.

  We view $\wh P_b^{(L)}$ as a polynomial in $z_L^2$ with coefficients in the other $z_i$'s and $w$, and to satisfy the degree we need to know the value of the polynomial for at least $2L+1$ values of $z_L^2$. As in the previous proof, we list here recursions which will give the value of $\wh P_b^{(L)}$ at $4(L-1)$ values of $z_L^2$. These recursions come from \propref{P0rec}.
  \[
  \wh P_b^{(L)}|_{z_L^2=(qz_k)^{\pm 2}}=\wh P_b^{(L-2)}(\hat z_k,\hat z_L),\qquad 1\leq k<L.
  \]
  The proof that these recursions are satisfied by the proposed expression for $\wh P_b^{(L)}$ is again based on the known recursions for $\chi$ \eqref{eq:chirecur}, but also relies on properties of $k(a,b)$ which imply that
  \[\frac{k(z_k,w)k(z_k,q/w)}{k(1/w,z_k)k(1/w,qz_k)}
  =\frac{k(z_k/q,w)k(z_k/q,w/q)}{k(1/w,z_k)k(1/w,z_k,q)}=1.\]

  It thus remains to show that the expression is true for $L=1$,
  $$
    \wh P_b^{(1)}
    \stackrel{?}{=} \frac{-[q][w^2]}{k(1/w,z_1)}
    \times
    \frac{\chi_{2}(w^2,z_1^2)\chi_{2}((q/w)^2,z_1^2)}{\chi_{1}(z_1^2)^2}
    =\frac{-[q][w^2]}{k(1/w,z_1)}.
  $$
  There are four configurations of $\wh\rho$, three of which contribute to $\wh P_b^{(1)}$. The one which does not has a weight
  $([qz_1/w][q/z_1w])([qw/z_1][qz_1w])^{-1}$, so
  \begin{align*}
    \wh P_b^{(1)} &= 1-\frac{[qz_1/w][q/z_1w]}{[qw/z_1][qz_1w]}
    = \frac{[qw/z_1][qz_1w]-[qz_1/w][q/z_1w]}{k(1/w,z_1)}
    = \frac{-[q][w^2]}{k(1/w,z_1)} \,. \qedhere
  \end{align*}
\end{proof}

\subsubsection{Homogeneous limit}
\label{sec:homresult}
In this section we will use the homogeneous limit $z_1=\ldots=z_L=1$, $w^2=-q$, which causes the two orientations of the lattice faces to each have a probability of $1/2$.

\begin{prop}
  The expression for $P_b^{(L)}$ in the homogeneous limit is
  \begin{equation}
    \label{eq:P1hom}
    P_b^{(L)}=\frac{A_V(L)A_V(L+2)}{N_8(L+1)^2} \,,
  \end{equation}
  where $A_V(L)$ and $N_8(L)$ are the number of $L\times L$ vertically symmetric alternating sign matrices and cyclically symmetric self-complementary plane partitions of size $L\times L\times L$ respectively, and have the explicit expressions
  \[ A_V(2m+1)=\prod_{i=0}^{m-1}\frac{(3i+2)(6i+3)!(2i+1)!}{(4i+2)!(4i+3)!}
  \,,
  \qquad
  N_8(2m)=\prod_{i=0}^{m-1}\frac{(3i+1)(6i)!(2i)!}{(4i)!(4i+1)!} \,. \]
\end{prop}

\begin{proof}
  The result~\eqref{eq:P1hom} is simply obtained from~\eqref{eq:P1form} and the homogeneous limit of the symplectic characters \cite{DF07},
  \begin{align*}
    \chi_{2m}(1,\ldots,1)&=3^{m(m-1)}A_V(2m+1) \,, \\
    \chi_{2m-1}(1,\ldots,1)&=3^{(m-1)^2}N_8(2m) \,. \qedhere
  \end{align*}
\end{proof}

\begin{prop}
  The expression for $\wh P_b^{(L)}$ in the homogeneous limit is
  \begin{equation}
    \label{eq:P0hom}
    \wh P_b^{(L)}=\frac{3}{4^{L}}\frac{A(L)^2}{N_8(L+1)^2A_V(L)^2} \,,
  \end{equation}
  where $A(L)$ is the number of $L\times L$ alternating sign matrices,
  $A(L)=\prod_{i=0}^{L-1}\frac{(3i+1)!}{(L+i)!}$.
\end{prop}

\begin{proof}
  For $\wh P_b^{(L)}$ in~\eqref{eq:P0form}, taking the homogeneous limit gives us
  \[
  \wh P_b^{(L)} = \frac{(-3)}{(-4)^{L}}
  \frac{\chi_{L+1}(-q,1,\ldots,1)^2}{3^{2(m-1)^2}N_8(L+1)^2},
  \]
  where we have set $L=2m-1$. To simplify the numerator, we use the relation \cite{ZJComm}
  $$
  s_{\lambda}^{(4m-2)}(u_1,u_1^{-1},\ldots,u_{2m-1},u_{2m-1}^{-1})
  =\chi_{2m}(q,u_1,\ldots,u_{2m-1}) \times \chi_{2m}(-q,u_1,\ldots,u_{2m-1}) \,,
  $$
  where $s_\lambda$ is the Schur function for the partition $\lambda$. The homogeneous expression for $s_\lambda$ is
  \[ s_{\lambda}^{(2L)}(1,\ldots,1)=3^{L(L-1)/2}A(L), \]
  and thus $\chi_{L+1}(-q,1,\ldots,1)$ becomes
  \[
  \chi_{L+1}(-q,1,\ldots,1)=\frac{3^{(2m-1)(m-1)}A(L)}{\chi_{L+1}(q,1,\ldots,1)}.
  \]
  We then use the recursion for $\chi_{L+1}(q,1\ldots,1)$ to get
  $$
    \chi_{L+1}(-q,1,\ldots,1)
    =\frac{3^{(2m-1)(m-1)}A(L)}{k(1,1)^{2m-2}\chi_{2m-2}(1,\ldots,1)} \\
    =\frac{3^{(m-1)^2}A(L)}{A_V(L)} \,,
  $$
  from which the result~\eqref{eq:P0hom} follows immediately.
\end{proof}

\subsubsection{Large-$L$ limit}
\label{sec:limresult}
\begin{prop}
  In the limit $L \to \infty$, $P_b^{(L)}$ and $\wh P_b^{(L)}$ as given in \eqref{eq:P1hom} and \eqref{eq:P0hom} have the asymptotic behaviour
  \begin{equation}
    \label{eq:Pb-asym}
    P_b^{(L)} \sim C \ L^{-1/3} \,,
    \qquad
    \wh P_b^{(L)} \sim \wh C \ L^{-1/3} \,,
  \end{equation}
  where
  \[
    C = \frac{9\times 2^{-5/3}\Gamma(1/3)\Gamma(5/6)}{\Gamma(1/6)\Gamma(2/3)} \,,
    \qquad
    \wh C = \frac{2^{14/9}\pi^{1/3}G(1/3)^4 G(5/6)^4}{G(1/2)^4 G(2/3)^2} \,.
  \]
\end{prop}

Hence, the lattice probabilities $P_b$ and $\wh P_b$ follow the limiting behaviour~\eqref{eq:schramm2} predicted by Schramm's formula at $\kappa=6$ with the correct exponent, but with different multiplicative constants. This is because $P_b$ and $\wh P_b$ correspond to fixing a position $j$ and letting $L$ tend to infinity, whereas the scaling limit would require a fixed ratio $x = \pi j/L$. In the following section, we address the determination of $P_{\rm left}$ for general $j$.

\section{Left-passage probabilities}
\label{sec:Pleft}

\subsection{Definitions and properties}

\begin{defn}
  For $j \in \{ 1, \dots, L\}$, we denote by $X_j(z_1, \dots, z_L)$ the probability that the path $\gamma$ passes through the $j$-th horizontal edge in a horizontal section of the type shown in \figref{path}a.
\end{defn}

\begin{defn}
  For $j \in \{ 0, \dots, L+1\}$, we denote by $\wh X_j(w;z_1, \dots, z_L)$ the probability that the path $\gamma$ passes through the $j$-th horizontal edge in a horizontal section of the type shown in \figref{path}b.
\end{defn}

\begin{defn}
  For $j \in \{ 1, \dots, L+1\}$, we denote by $Y_j(w;z_1, \dots, z_L)$ the probability that the path $\gamma$ passes through the $j$-th vertical edge above a horizontal section of the type shown in \figref{path}a.\symbolfootnote{We could also define $\wh Y_j$ in a similar fashion, but this is simply related to $Y_j$ by the transformation $z_k\rightarrow 1/z_k$, $\forall k$.}
\end{defn}

\begin{lm}
  \label{lm:PX}
  Using the convention that the vertices on a horizontal section like in \figref{path}a are numbered  $\{1/2, \dots, L+1/2\}$, the 
  probability that $\gamma$ passes to the left of $(j+1/2)$ is
  $$ P_{j+1/2} = \sum_{\ell=1}^j (-1)^{\ell-1} X_\ell \,. $$
  On a horizontal section like in \figref{path}b, the probability that $\gamma$ passes to the left of $(j+1/2)$ is
  $$ \wh P_{j+1/2} = \sum_{\ell=0}^j (-1)^{\ell} \wh X_\ell \,. $$
  In particular, we have $X_1= P_{3/2}= P_b$ and $\wh{X}_0= \wh{P}_{1/2}= Y_1=\wh P_b$.
\end{lm}

The power of $(-1)$ in the above sums can be explained in the following way. The infinite loop, oriented as in \figref{path}, can only pass in one direction through any given edge: For $X_j$, it passes upwards if $j$ is odd and downwards if $j$ is even; for $\wh X_j$ the opposite is true. When calculating the left-passage probabilities all configurations with the path passing downwards through the edge must be counted with a negative sign, thus we arrive at the above expressions.

\begin{lm}
  The probabilities $X$, $\wh X$ and $Y$ are related by the conservation property
  \[
    \forall j \in \{ 1, \dots, L\} \,, \quad X_j + \wh{X}_j = Y_j+Y_{j+1} \,,
  \]
  and the identities for special values of $w$
  \begin{align*}
    \left. Y_j \right|_{w=z_j} &= \left. \wh{X}_j \right|_{w=z_j} \,, & \left. Y_{j+1} \right|_{w=z_j} &= X_j \,, \\
    \left. Y_j \right|_{w=qz_j} &= X_j \,, & \left. Y_{j+1} \right|_{w=qz_j} &= \left. \wh{X}_j \right|_{w=qz_j} \,.
  \end{align*}
\end{lm}
\begin{proof}
The conservation property comes from considering the possible configurations on a face. Each $X$ and $Y$ has two terms corresponding to the two possible configurations, and for each $X$ one of these terms matches one of those for one of the $Y$s. At the special values $w=z_j$ and $w=qz_j$, the face at position $j$ is specialised to one of the configurations, and the identities follow immediately.
\end{proof}

Similarly to $P_b$ and $\wh P_b$, the probabilities $X$, $\wh X$ and $Y$ satisfy some symmetry and recursion relations.
\begin{prop}
  The probabilities $X_j(z_1,\dots,z_L)$ and $\wh X_j(w; z_1,\dots,z_L)$ are symmetric functions of $\{z_1, \dots, z_{j-1}\}$ and $\{z_{j+1}, \dots, z_L\}$ separately. $X_j(z_1,\dots,z_L)$ is invariant under $z_\ell \to 1/z_\ell$ for all $\ell$, and $\wh X_j(w; z_1,\dots,z_L)$ is invariant under $z_\ell \to 1/z_\ell$ for $\ell \neq j$. The probability $Y_j(w; z_1,\dots,z_L)$ is a symmetric function of $\{z_1, \dots, z_{j-1}\}$ and $\{z_j, \dots, z_L\}$ separately, and invariant under $z_\ell \to 1/z_\ell$ for all $\ell$.
\end{prop}
\begin{proof}
  The proof is similar to the ones for \propref{P1symm} and \propref{P0symm}.
\end{proof}

\begin{prop}
  The probabilities $X_j(z_1,\dots,z_L)$ and $\wh X_j(w; z_1,\dots,z_L)$ satisfy the recursion relations
  \begin{align*}
    \left. X_j^{(L)} \right|_{z_1^2 = (qz_\ell)^{\pm 2}} &= X_{j-2}^{(L-2)}(\hat z_1,\hat z_\ell)\,, & \left. \wh X_j^{(L)} \right|_{z_1^2 = (qz_\ell)^{\pm 2}} &= \wh X_{j-2}^{(L-2)}(\hat z_1,\hat z_\ell)\,, & &\text{for $1<\ell<j$,}\\
    \left. X_j^{(L)} \right|_{z_L^2 = (qz_\ell)^{\pm 2}} &= X_j^{(L-2)}(\hat z_\ell,\hat z_L)\,, & \left. \wh X_j^{(L)} \right|_{z_L^2 = (qz_\ell)^{\pm 2}} &= \wh X_j^{(L-2)}(\hat z_\ell,\hat z_L)\,, & &\text{for $j<\ell<L$,}
  \end{align*}
  whereas $Y_j(w; z_1,\dots,z_L)$ satisfies
  \begin{align*}
    \left. Y_j^{(L)} \right|_{z_1^2 = (qz_\ell)^{\pm 2}} &= Y_{j-2}^{(L-2)}(\hat z_1,\hat z_\ell) \,, \qquad \text{for $1<\ell<j$,}\\
    \left. Y_j^{(L)} \right|_{z_L^2 = (qz_\ell)^{\pm 2}} &= Y_j^{(L-2)}(\hat z_\ell,\hat z_L) \,, \qquad \text{for $j\leq \ell< L$.}
  \end{align*}
\end{prop}
\begin{proof}
  The proof is similar to the ones for \propref{P1rec} and \propref{P0rec}.
\end{proof}

The above relations satisfied by $X_j$, $\wh X_j$ and $Y_j$ are, in principle, sufficent to determine completely these quantities. However, they turn out to be particularly difficult to solve in practice, and we resort to different tools to describe them.

\subsection{Numerical study}

For a given choice of the $z_j$'s, one can find numerically the components $\psi_\alpha$ by the power method, and use these to evaluate $X_j$, $\wh X_j$ and $Y_j$. The left-passage probabilities are then obtained from \lmref{PX}.

For small enough system sizes, we find the left-passage probabilities $P_{j+1/2}$ as rational fractions (see \tabref{Pj}).

\begin{table}[ht]
  \begin{center}
    \begin{tabular}{|r|r|l|}
      \hline
      $L$ & $Z_L$ & $P_j^{(L)} \times Z_L^2$ \\
      \hline
      3 & 2 & 0, 3, 1, 4 \\
      5 & 11 & 0, 78, 22, 99, 43, 121 \\
      7 & 170 & 0, 16796, 4484, 21093, 7807, 24416, 12104, 28900 \\
      9 & 7429 & 0, 29641710, 7721790, 37074705, 12859293 \dots \\
      11 & 920460 & 0, 426943865250, 109785565350, 532943651700, 178807268772,
      605036201854 \dots \\
      \hline
    \end{tabular}
  \end{center}
  \caption{Left-passage probabilities $P_j^{(L)}$ for strips of width $L\leq
    11$. Data obtained by numerical diagonalisation of the transfer matrix.}
  \label{tab:Pj}
\end{table}

We now turn to the convergence of $P_{j+1/2}$ to Schramm's formula~\eqref{eq:schramm}. From \lmref{PX}, we see that $P_{j+1/2}$ is the sum of an alternating sequence, and has oscillations of wavelength $\delta j=1$. This phenomenon appears clearly in~\figref{pl-L21}. For this reason, we define the smooth and oscillatory parts as
$$
\overline{P}_j := \half \left( P_{j-1/2} + P_{j+1/2} \right) \,,
\qquad
\wt{P}_j := \half \left( P_{j-1/2} - P_{j+1/2} \right) = \half (-1)^{j-1} X_j \,.
$$
\begin{figure}[ht]
  \begin{center}
    \includegraphics{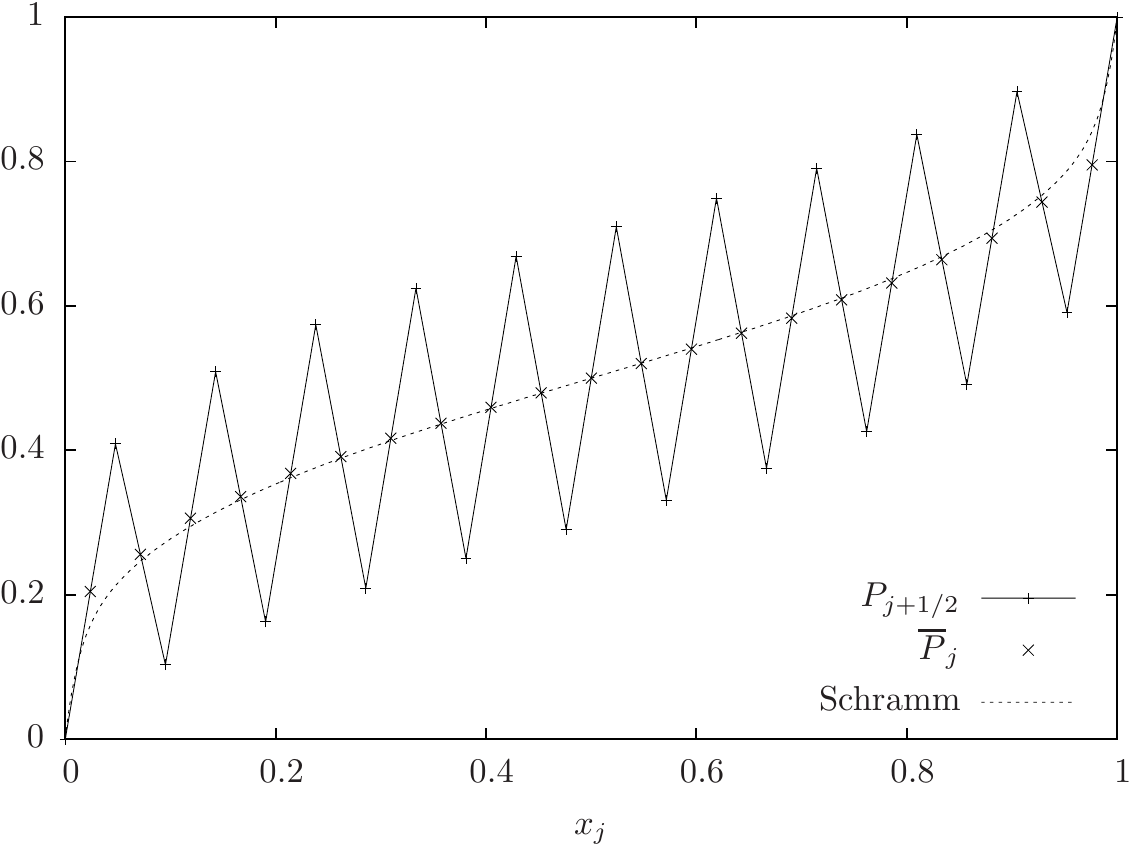}
  \end{center}
  \caption{The left-passage probability $P_{j+1/2}$ and its smooth part
    $\overline{P}_j$ for $L=21$, compared to Schramm's formula~\eqref{eq:schramm}.}
  \label{fig:pl-L21}
\end{figure}
In~\figref{pl-L21} we see that $\overline{P}_j$ is very close to Schramm's formula for $L=21$. In \figref{pl}, we compare the data for $\overline{P}_j$ at various system sizes, and observe very good convergence to Schramm's formula. Finite-size effects are more important near the boundaries, but as we already noted in \secref{Pb}, the scaling of $P_{1/2}$ with $L$ is the one predicted by Schramm's formula.

Finally, in \figref{pl-osc}, we plot the oscillatory part $\wt{P}_j$. This quantity is a lattice effect, and is not predicted directly by Schramm's formula.

\begin{figure}
  \begin{center}
    \includegraphics{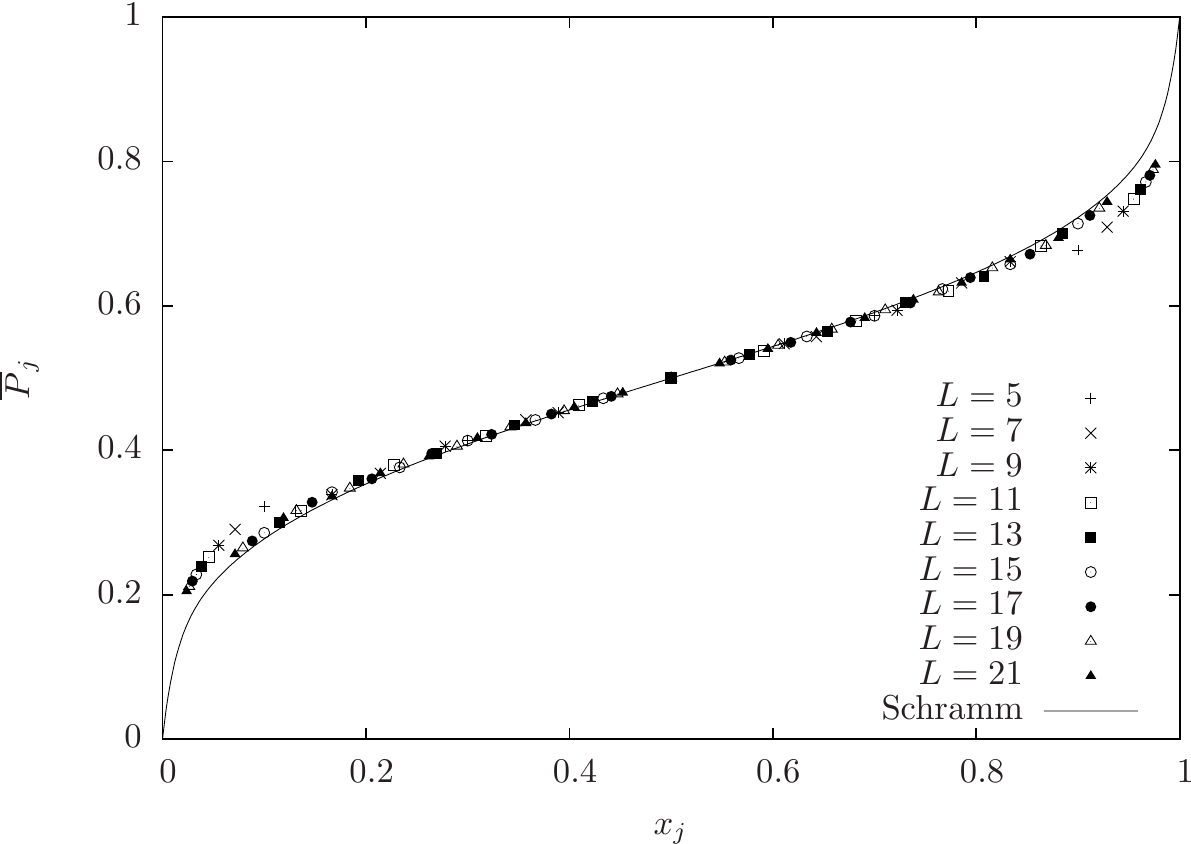}
  \end{center}
  \caption{Smooth part of the left-passage probability, compared to 
    Schramm's formula~\eqref{eq:schramm}.}
  \label{fig:pl}
\end{figure}

\begin{figure}
  \begin{center}
    \includegraphics{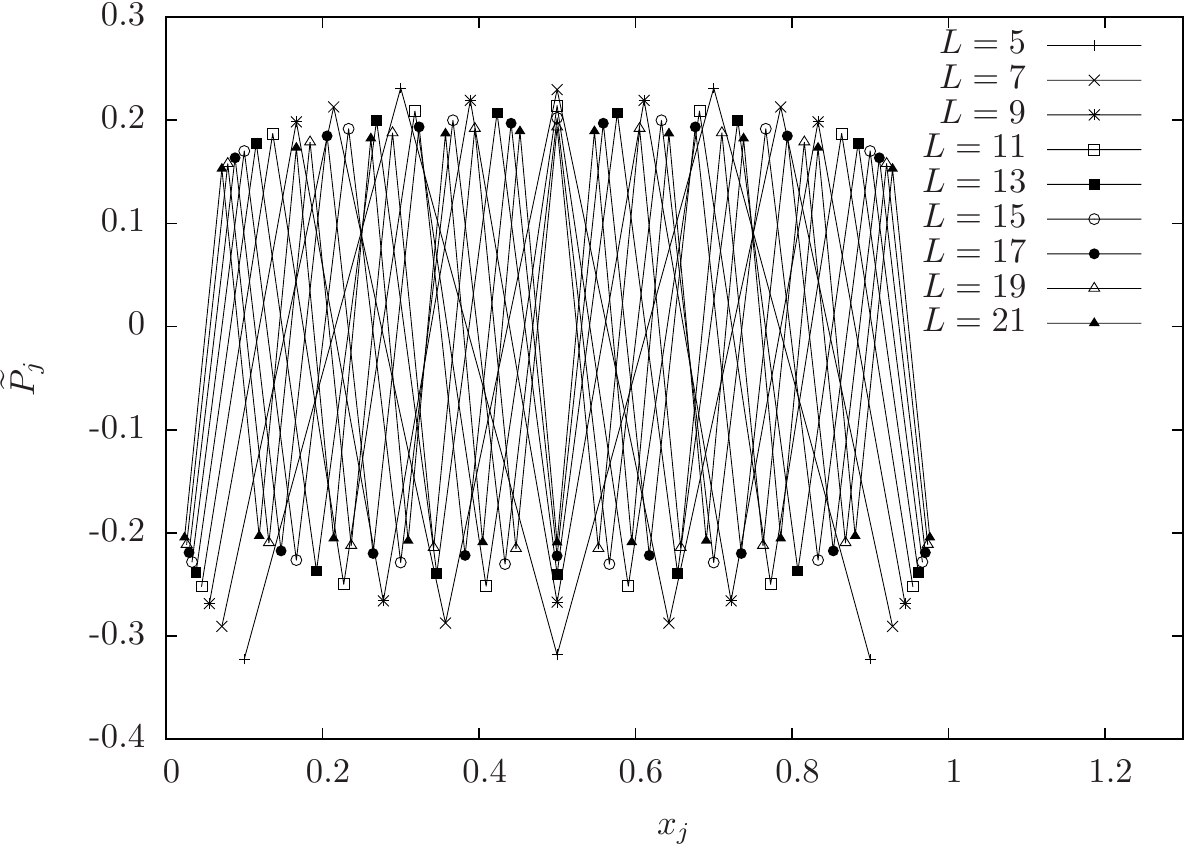}
  \end{center}
  \caption{Oscillatory part of the left-passage probability.}
  \label{fig:pl-osc}
\end{figure}

\subsection{Left-passage probability in the FK model}
\label{sec:FK}
\subsubsection{Mapping to the six-vertex model}

Throughout this section, we remove the restriction on $q$, and we use the algebraic Bethe ansatz notations
\[
  q = e^\eta \,, \qquad
  z_j = e^{-v_j} \,, \qquad
  w = e^{-u} \,.
\]
Moreover, we introduce for convenience $\eta'=\eta+i\pi$, so that
the loop weight reads
$$
n = -(q+q^{-1}) = -2 \cosh \eta = 2 \cosh \eta' \,.
$$
\begin{figure}[ht]
  \begin{center}
    \includegraphics{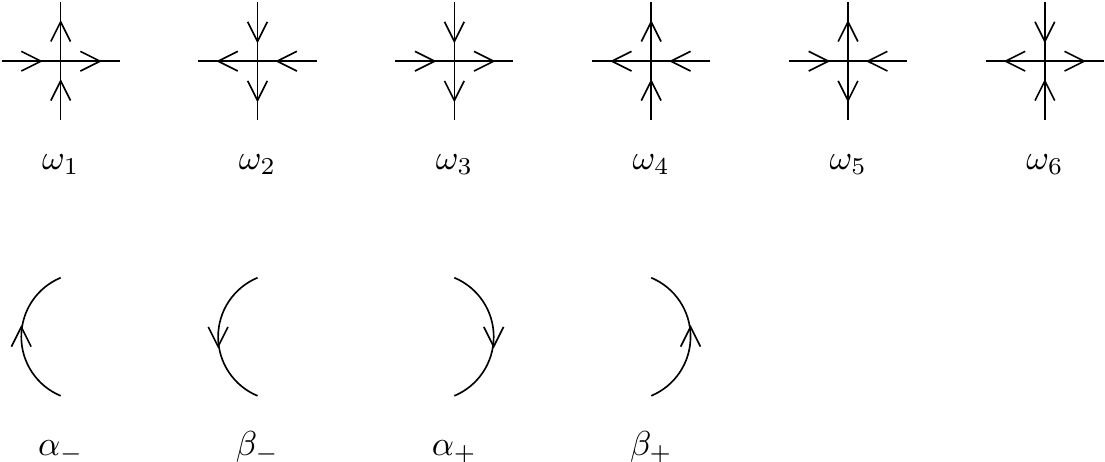}
    \caption{The configurations of the six-vertex model and the associated
      Boltzmann weights.}
    \label{fig:6V}
  \end{center}
\end{figure}

Following~\cite{Bax82}, we distribute the loop weight locally by orienting the loops: to each $\pi/2$ left (resp. right) turn of a loop, we associate a phase factor $e^{\eta'/4}$ (resp. $e^{-\eta'/4}$) in the Boltzmann weights. Forgetting about the loop connectivities results in a six-vertex (6V) model (see \figref{6V}). Using the Boltzmann weights defined by the $R$-matrix in \defnref{R} (with $z=1$ wlog for the moment), the 6V weights can be rescaled to
\begin{equation} \label{eq:6V-bulk}
  \begin{array}{rcl}
  \omega_1=\omega_2 &=& \sinh(\eta+u) \\
  \omega_3=\omega_4 &=& \sinh u \\
  \omega_5 &=& e^{+\eta'/2+u} \sinh \eta \\
  \omega_6 &=& e^{-\eta'/2-u} \sinh \eta \,.
  \end{array}
\end{equation}
Moreover, on the boundary, the loops undergo a half-turn, and hence the boundary weights must be
\begin{equation} \label{eq:6V-bound}
  \alpha_\pm = e^{-\eta'/2} \,, \quad \beta_\pm = e^{+\eta'/2} \,.
\end{equation}
Hence the 6V model resulting from this mapping is described by the matrices
\begin{equation*}
  \ul{R}(u) = \left(\begin{array}{cccc}
    \sinh(\eta+u) & 0 & 0 & 0 \\
    0 & \sinh u & e^{-\frac{\eta'}{2}-u}\sinh \eta & 0 \\
    0 & e^{\frac{\eta'}{2}+u}\sinh \eta & \sinh u & 0 \\
    0 & 0 & 0 & \sinh(\eta+u)
  \end{array}\right) \,,
  \qquad 
  \ul{K}_\pm = \left(\begin{array}{cc}
    e^{-\frac{\eta'}{2}} & 0 \\
    0 & e^{\frac{\eta'}{2}}
  \end{array}\right) \,,
\end{equation*}
and the corresponding transfer matrix is denoted by $\ul{t}_{\rm 6V}$. These matrices are related to the standard 6V ones by the ``gauge transformation'':
\begin{equation}
  \label{eq:gauge}
  \begin{array}{rcl}
    \ul{R}_{ab}(u-v) &=&
    e^{-(\frac{u}{2}+\frac{\eta'}{4})\sigma_a^z -\frac{v}{2} \sigma_b^z}
    \ R_{ab}(u-v)
    \ e^{(\frac{u}{2}+\frac{\eta'}{4})\sigma_a^z +\frac{v}{2} \sigma_b^z} \\
    &=&
    e^{-\frac{u}{2}\sigma_a^z -(\frac{v}{2}-\frac{\eta'}{4}) \sigma_b^z}
    \ R_{ab}(u-v)
    \ e^{\frac{u}{2}\sigma_a^z +(\frac{v}{2}-\frac{\eta'}{4}) \sigma_b^z} \,,
    \\
    \\
    \ul{K}_\pm &=& e^{\pm(u/2+\eta'/4)\sigma^z} \ K_\pm(u)
    \ e^{\pm(u/2+\eta'/4)\sigma^z} \,,
  \end{array}
\end{equation}
where $K_+(u) = 2 e^{\xi_+}K(u+\eta,\xi_+)$, $K_-(u) = 2 e^{\xi_-}K(u,\xi_-)$, and
\begin{equation*}
  R(u) = \left(\begin{array}{cccc}
    \sinh(\eta+u) & 0 & 0 & 0 \\
    0 & \sinh u & \sinh \eta & 0 \\
    0 & \sinh \eta & \sinh u & 0 \\
    0 & 0 & 0 & \sinh(\eta+u)
  \end{array}\right) \,,
  \qquad
  K(u,\xi) = \left(\begin{array}{cc}
    \sinh(\xi+u) & 0 \\
    0 & \sinh(\xi-u)
  \end{array}\right) \,,
\end{equation*}
with the values of the boundary parameters:
$\xi_\pm = \mp \infty$.
It is customary to shift the spectral parameters to define the monodromy matrices~\cite{Skl88}:
$$
u := \lambda-\eta/2 \,,
\qquad
v_j := \xi_j-\eta/2 \,,
$$
and we thus write:
\begin{equation}
  \begin{array}{rcl}
    T(\lambda) &:=& R_{0L}(\lambda-\xi_L) \dots R_{01}(\lambda-\xi_1) \\
    \wh{T}(\lambda) &:=& R_{10}(\lambda+\xi_1-\eta) \dots
    R_{L0}(\lambda+\xi_L-\eta) \\
    t_{\rm 6V}(\lambda) &:=& {\rm Tr_0} \left[
      K_+(\lambda) T(\lambda) K_-(\lambda) \wh{T}(\lambda)
      \right] \,.
  \end{array}
\end{equation}
One can easily show that the transfer matrices before and after the gauge change~\eqref{eq:gauge} are simply related by a similarity transformation
\begin{equation}
  \ul{t}_{\rm 6V} = G^{-1} \ t_{\rm 6V} \ G \,,
  \qquad \text{where} \qquad
  G := \prod_{j=1}^L e^{v_j \sigma_j^z /2} \,.
\end{equation}
If we specialise to a homogeneous system where all the $\xi_j$'s are set to $\eta/2$, the very anisotropic limit $\lambda \to \eta/2$ yields the open XXZ Hamiltonian
\begin{eqnarray}
  {\cal H}_{\rm XXZ} &:=&
  \left.\frac{\partial \log t_{\rm 6V}(\lambda)}{\partial
    \lambda}\right|_{\lambda=\eta/2} \nn \\
  &=& \sum_{j=1}^{L-1} \left[
  \sigma_j^x \sigma_{j+1}^x + \sigma_j^y \sigma_{j+1}^y
  + \cosh \eta \ \sigma_j^z \sigma_{j+1}^z
  \right]
  - \sinh \eta \ (\sigma_1^z-\sigma_L^z) \,.
  \label{eq:HXXZ}
\end{eqnarray}
Note that in the critical regime ($\eta \in i\mathbb{R}$), the boundary terms are imaginary. In the remainder of this section, we shall restrict ourselves to the homogeneous system described by ${\cal H}_{\rm XXZ}$, but our results can be readily generalised to an arbitrary choice of the $\xi_j$'s.

\subsubsection{The left-passage probability as an XXZ correlation function}

\begin{prop}
  In the critical regime $\eta \in i\mathbb{R}$, the following identity holds:
  \begin{equation}
    P_{j+1/2} = \sum_{\ell=1}^j {\rm Re} \left(\frac{\bra{\Psi_0}\sigma_\ell^z\ket{\Psi_0}}
    {\langle \Psi_0 \ket{\Psi_0}} \right) \,,
  \end{equation}
  where $\bra{\Psi_0}$ and $\ket{\Psi_0}$ are the left and right eigenvectors of ${\cal H}_{\rm XXZ}$~\eqref{eq:HXXZ} associated to the lowest energy.
\end{prop}

\begin{proof}
  We have used above the mapping of TL loop configurations to 6V arrow configurations, through the orientation of loops. Consider the intermediate model, i.e., the oriented TL (OTL) loops. In this model, a configuration $C$ on the whole lattice has a Boltzmann weight (at the isotropic point $\lambda=-i\pi/2$)
  $$
    W[C] = (-q)^{\# {\rm anti-clockwise\ loops}(C)}
    \times (-1/q)^{\# {\rm clockwise\ loops}(C)} \,,
    $$
    and the partition function is equal to the one of the original TL model
    $$
    {\cal Z}_{\rm OTL} = \sum_{{\rm oriented\ config.}\ C} W[C]
    = \sum_{{\rm unoriented\ config.}\ C} (-q-q^{-1})^{\# {\rm loops}(C)}
    = {\cal Z}_{\rm TL} \,.
  $$
   
  \begin{figure}[ht]
    \begin{center}
      \includegraphics[scale=0.75]{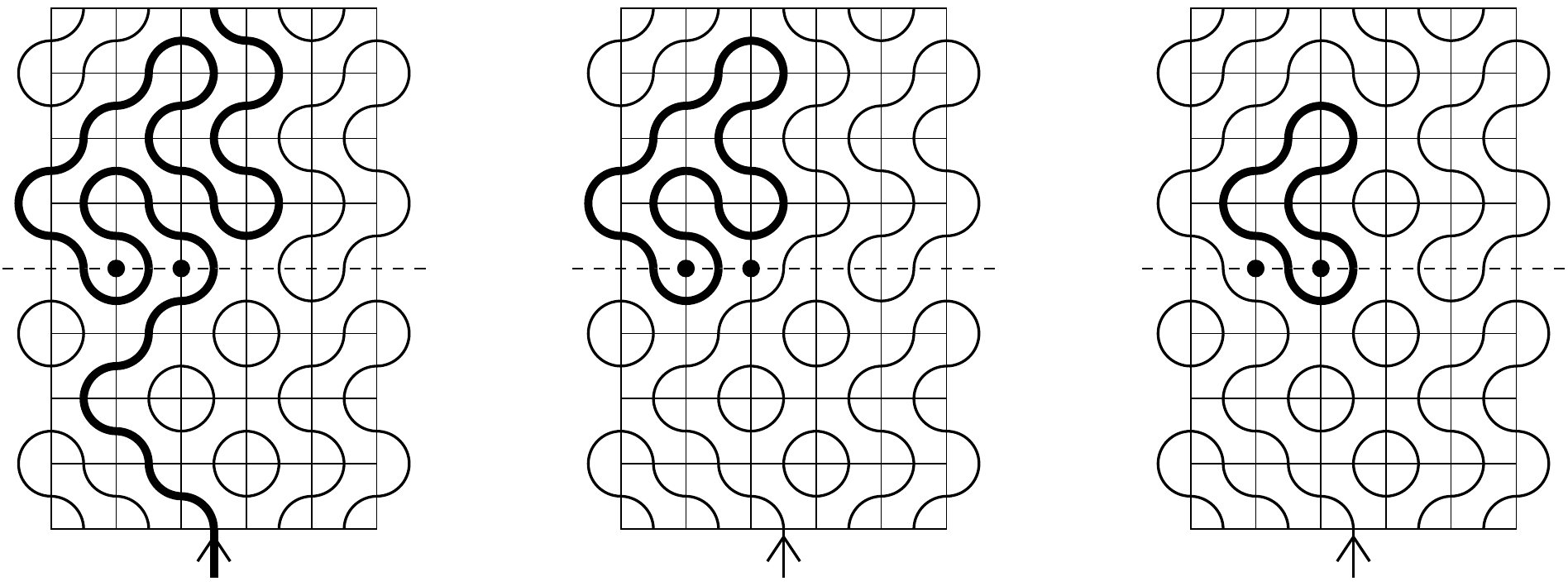}
      \caption{From left to right: loop configurations belonging to the
        subsets ${\cal A}_j$, ${\cal A}'_j$, and ${\cal A}''_j$, for $j=2$
        and $L=5$.}
      \label{fig:path2}
    \end{center}
  \end{figure}

  We now consider a horizontal section of the strip, of the type shown in \figref{path}a. For a given oriented loop configuration $C$, we denote $\sigma_j^z(C) \in \{1,-1\}$ the orientation of the arrow across the $j$-th horizontal edge. There are three possibilities for the loop passing through this edge (see \figref{path2}):
  \begin{itemize}
  \item The loop is the open path $\gamma$. Then $\sigma_j^z(C)=(-1)^{j-1}$.

  \item The loop encloses the marked point on the left of $j$. Then
    $\sigma_j^z(C)=+1$ iff the loop is oriented anti-clockwise.

  \item The loop encloses the marked point on the right of $j$. Then
    $\sigma_j^z(C)=+1$ iff the loop is oriented clockwise.
  \end{itemize}
  We denote by ${\cal A}_j, {\cal A}'_j, {\cal A}''_j$ the corresponding subsets
  of oriented loop configurations.
  The expectation value of $\sigma_j^z$ in the OTL model then reads
  \begin{eqnarray}
    \langle \sigma_j^z \rangle_{\rm OTL} &=& \frac{1}{\cal Z}
    \left[
      (-1)^{j-1} \ \sum_{C \in {\cal A}_j}
      + \tanh \eta
      \ \left(\sum_{C \in {\cal A}'_j}-\sum_{C \in
        {\cal A}''_j} \right)
      \right] W[C] \nn \\
    &=& (-1)^{j-1} \ X_j + \tanh \eta \ (\Pb[C \in {\cal A}'_j]-\Pb[C \in {\cal A}''_j]) \,.
    \label{eq:sigmaz}
  \end{eqnarray}
  Finally, from the mapping described above, we have
  \begin{align*}
    \frac{\bra{\Psi_0}\sigma_j^z\ket{\Psi_0}}
         {\langle \Psi_0 \ket{\Psi_0}}
         &= \langle \sigma_j^z \rangle_{\rm OTL} \,.
         \qedhere
  \end{align*}
\end{proof}

\section{Perspectives}
\label{sec:conclusion}

In the percolation model, we have obtained an exact expression
for the boundary passage probabilities $P_b$ and $\wh{P}_b$,
which are lattice analogs of a boundary observable in SLE$_6$.
In the FK cluster model with generic $Q$, we have related the
left-passage probability to the magnetisation in a solvable open
XXZ spin chain.

Our results have many possible developments. First, within the $q$KZ
approach, we hope to exploit the symmetry and recursion relations
for the probabilities $X_j$, $\wh{X}_j$ and $Y_j$ to find their
explicit expression. This certainly involves a deeper understanding of the
properties of symplectic characters and Schur functions~\cite{dGP11}.
Second, with the algebraic Bethe ansatz~\cite{KitanineKMNST07,KitanineKMNST08}, it seems possible
to calculate
the magnetisation $\langle \sigma_j^z \rangle$ in the open XXZ chain in a
closed form, at least for $j$ close enough to one of the boundaries.
The advantages of this method is that it is valid for any value of the
deformation parameter $q$, and, in cases where a closed form cannot be
achieved, it still produces determinant forms which can be evaluated
numerically for very large system sizes ($L \sim 1000$ sites).
Finally, we note that the probabilities $X_j$, $\wh{X}_j$ and $Y_j$
are very similar to the discretely holomorphic parafermions found
for the TL loop model~\cite{IkhlefC09,RivaC06}, which are the starting point
in the proof of conformal invariance for the Ising model~\cite{ChelkakS09,Smir07}.
Thus the study of these objects on a general domain $\Omega$ may
bring progress in extending this proof to the FK model with generic $Q$.

\section*{Acknowledgments}

The authors thank J.~de Gier, C.~Hagendorf, A.~Mays, V.~Terras, and P.~Zinn-Justin
for fruitful discussions. This work was supported by the European Research
Council (grant CONFRA 228046).


\end{document}